\newif\ifarxiv
\arxivtrue

\ifarxiv
\documentclass{article}
\usepackage[utf8]{inputenc}
\usepackage{geometry, amsmath, amsthm, algorithm, algpseudocode, amssymb, natbib}
\usepackage{fullpage}
\usepackage{authblk}
\else
\documentclass[format=acmsmall, review=false]{acmart}
\usepackage{acm-ec-24-proc}
\usepackage{booktabs} 
\usepackage[ruled]{algorithm2e} 
\usepackage{natbib} 

\SetAlFnt{\small}
\SetAlCapFnt{\small}
\SetAlCapNameFnt{\small}
\SetAlCapHSkip{0pt}
\IncMargin{-\parindent}
\setcitestyle{authoryear}



\acmYear{2024}\copyrightyear{2024}
\setcopyright{acmlicensed}
\acmConference[EC '24]{Conference on Economics and Computation}{July 8--11, 2024}{New Haven, CT, USA}
\acmBooktitle{Conference on Economics and Computation (EC '24), July 8--11, 2024, New Haven, CT, USA}
\acmDOI{10.1145/3670865.3673622}
\acmISBN{979-8-4007-0704-9/24/07}

\title[Forecasting for Swap Regret for All Downstream Agents]{Forecasting for Swap Regret for All Downstream Agents}
\author{Aaron Roth}
\email{aaroth@cis.upenn.edu}
\orcid{0000-0002-0586-0515}
\affiliation{
  \institution{University of Pennsylvania}
  \city{Philadelphia}
  \state{PA}
  \country{USA}
}
\author{Mirah Shi}
\email{mirahshi@seas.upenn.edu}
\orcid{0009-0000-2162-4945}
\affiliation{
  \institution{University of Pennsylvania}
  \city{Philadelphia}
  \state{PA}
  \country{USA}
}
\fi

\usepackage{bbm}
\usepackage{mathtools}
\allowdisplaybreaks
\usepackage{xcolor}
\usepackage{graphicx} 
\usepackage{hyperref}

\DeclareMathOperator*{\E}{\mathbb{E}}
\newcommand{\eps}{\varepsilon}
\newcommand{\1}{\mathbbm{1}}
\newcommand{\cA}{\mathcal{A}}
\newcommand{\cC}{\mathcal{C}}
\newcommand{\cE}{\mathcal{E}}

\newcommand{\cU}{\mathcal{U}}
\newcommand{\cV}{\mathcal{V}}
\newcommand{\cP}{\mathcal{P}}

\newcommand{\R}{\mathbb{R}}
\newcommand{\<}{\langle}
\renewcommand{\>}{\rangle}
\DeclareMathOperator*{\argmax}{arg\,max}
\DeclareMathOperator*{\argmin}{arg\,min}
\DeclarePairedDelimiter\floor{\lfloor}{\rfloor}

\newtheorem{definition}{Definition}

\newtheorem{lemma}{Lemma}
\newtheorem{theorem}{Theorem}
\newtheorem{remark}{Remark}

\newtheorem{observation}{Observation}
\newtheorem{assumption}{Assumption}

\newenvironment{proofof}[1]{\par{\noindent \textit{Proof of #1.}}}{\hspace*{\fill} $\qed$ \par}

\begin{document}

\ifarxiv
\title{Forecasting for Swap Regret for All Downstream Agents}
\author{Aaron Roth}
\author{Mirah Shi}
\affil{Department of Computer and Information Sciences, University of Pennsylvania}
\fi

\ifarxiv
\maketitle
\fi

\begin{abstract}
We study the problem of making predictions so that downstream agents who best respond to them will be guaranteed diminishing swap regret, \emph{no matter what their utility functions are}. It has been known since \citet{foster1997calibrated} that agents who best-respond to calibrated forecasts have no swap regret. Unfortunately, the best known algorithms for guaranteeing calibrated forecasts in sequential adversarial environments do so at rates that degrade exponentially with the dimension of the prediction space (e.g. the number of actions). Thus the calibration approach is substantially worse, in terms of swap regret \emph{rates}, compared to the approach of running a bespoke swap-regret-minimization algorithm for each agent --- an approach that requires knowing the agent utility functions. In this work, we show that by making predictions that are not calibrated, but are unbiased subject to a carefully selected collection of events,  we can guarantee arbitrary downstream agents diminishing swap regret at rates that substantially improve over the rates that result from calibrated forecasts --- while maintaining the appealing property that our forecasts give guarantees for \emph{any} downstream agent, without our forecasting algorithm needing to know their utility function. 

We give separate results in the ``low'' (1 or 2) dimensional setting and the ``high'' ($> 2$) dimensional setting. In the low dimensional setting, we show how to make predictions such that all agents who best respond to our predictions have diminishing swap regret --- in $1$ dimension, at the optimal $O(\sqrt{T})$ rate. In the high dimensional setting we show how to make forecasts that guarantee regret scaling at a rate of $O(T^{2/3})$ (crucially, a dimension independent exponent), under the assumption that downstream agents \emph{smoothly} best respond, and at the optimal rate of $O(\sqrt{T})$ under a perhaps less realistic behavioral assumption. Our results stand in contrast to rates that derive from agents who best respond to calibrated forecasts, which have an exponential dependence on the dimension of the prediction space. 
\end{abstract}

\ifarxiv
\else
\maketitle
\fi

\section{Introduction}
When acting in an adversarial environment, a popular way to evaluate the performance of a decision-making agent is through the lens of \emph{regret}: Given the sequence of outcomes actually realized, and the sequence of actions chosen by the decision making agent, how much better could they have done had they followed some other policy rather than what the actually did? External regret compares the agent's utility to a very simple benchmark: the utility they could have obtained had they played a \emph{constant} policy, that played the same fixed action at every round. Swap regret compares the agent's utility to a more demanding benchmark: the utility they could have obtained if, counterfactually, they had been able to go back and modify their play using an arbitrary \emph{swap function} mapping actions to replacement actions in any consistent fashion. Swap regret is closely linked to correlated equilibrium --- in particular, if all agents in a game have no swap regret, then the empirical distribution over their play is a correlated equilibrium \citep{foster1997calibrated}.

There are efficient algorithms that individual agents can use to guarantee themselves no swap regret in adversarial environments (e.g. \citet{blum2007internal}). But these algorithms require some sophistication to run, and must be run separately for each agent, as they are dependent on the agent's utility function. It would be compelling if it were possible to make public forecasts of the payoff relevant state, such that less sophisticated agents could simply best respond to these forecasts (as if they were correct), and be guaranteed that they will have no swap regret. These forecasts would be valuable universally to all agents, independently of what their utility functions are.  

A candidate solution to this problem is \emph{calibration}, a ``self-consistency'' measure of forecast quality. Informally, calibrated predictions must be unbiased conditional on the value of the prediction itself. It is possible to make predictions in adversarial environments that satisfy the calibration condition \citep{foster1998asymptotic}, and agents who best respond to calibrated predictions obtain diminishing swap regret \citep{foster1997calibrated}. However, this approach suffers from a serious shortcoming. Even in $1$ dimension, the best known algorithms obtain calibration (and hence downstream swap regret) at a rate of $T^{2/3}$ \citep{foster1998asymptotic,okoroafor2023faster} --- worse than the $O(\sqrt{T})$ rates known to be optimal for swap regret \citep{blum2007internal}. In fact, it is known that obtaining $\sqrt{T}$ rates via calibration is impossible even in $1$ dimension: \citet{qiao2021stronger} show an $\Omega(T^{0.528})$ lower bound on $1$-dimensional calibration rates  in adversarial environments. The situation only gets worse in higher dimensions --- the best known rates 
for $d$-dimensional forecasts scale as $T^{2d/(2d+1)}$ \citep{foster1997calibrated} (also see \citet{roth2023learning}). 
A related notion of \emph{smooth calibration} \citep{kakade2008deterministic,foster2018smooth} has similar guarantees under the assumption that downstream agents smoothly best respond (and has other desirable properties), but is not known to be obtainable at faster rates than traditional calibration. This suggests that we must look elsewhere if we wish to provide forecasts which can guarantee arbitrary downstream agents strong regret guarantees at reasonable rates.

Recent work of \citet{kleinberg2023ucalibration} accomplishes exactly this. They abandon swap regret, and consider the simpler goal of guaranteeing downstream agents low \emph{external} regret.  They show that a weaker measure of forecast quality called ``U-calibration''---corresponding to minimizing regret over a restricted collection of proper scoring rules---suffices for external regret minimization for arbitrary downstream agents. Remarkably, they show how to make predictions that guarantee diminishing external regret at a rate of $O(\sqrt{T})$ simultaneously for all downstream agents, thus bypassing the calibration lower bound. Given these results, one might wonder if sidestepping the shortcomings of calibration requires giving up on swap regret and settling for external regret. Put another way, is calibration necessary to achieve the \emph{stricter} goal of guaranteeing low \emph{swap} regret for all downstream agents?

\subsection{Our Contributions}

In this paper, we show that calibration is not necessary for this goal. Recall that predictions are calibrated if they are unbiased conditional on the value of their own prediction. We show that by requiring our predictions to be unbiased over an alternative collection of events, we can obtain substantially improved swap regret bounds simultaneously for all possible agents. For a single agent, we know that they will obtain no swap regret if they best respond to forecasts that are unbiased conditional on events defined by their own best response correspondence \citep{haghtalab2023calibrated,noarov2023highdimensional}, a weaker condition than full calibration---but one that depends on their utility function. At a high level, we take advantage of the structure of best response functions to define a collection of events such that forecasts that are unbiased with respect to these events suffice to guarantee no swap regret for all downstream agents. Using the algorithm of \citet{noarov2023highdimensional} we are able to produce forecasts that achieve this bias condition at rates that improve dramatically over the best known rates for calibration. 

In the $d=1$ dimensional setting we make forecasts that guarantee every best-responding downstream agent diminishing swap regret at the optimal rate of $\Tilde{O}(\sqrt{T})$ (Theorem \ref{thm:1d}). We can extend this technique to the $d=2$ dimensional setting while obtaining swap regret rates of $\Tilde{O}(T^{5/8})$  (Theorem \ref{thm:2d}). These results require no knowledge of either agent utility functions or the number of actions available to the agents.

We then turn to the higher dimensional setting, which presents new challenges (informally, that the structure of arbitrary best response functions projected onto a finite grid of predictions appears to become much more complex). In this higher $(d > 2)$ dimensional setting, we take a different approach, which requires two modifications to our setting. First, although we continue to assume no knowledge of the downstream agent's utility function, we now assume that we know an upper bound $k$ on the number of actions that they have available to them. Second, rather than assuming that agents exactly best respond to our forecasts, we assume that they \emph{smoothly} best respond --- informally that they choose an approximate best response using a mapping that is Lipchitz in our predictions. Quantal response e.g. satisfies this assumption; as we discuss, this is a behavioural assumption with extensive roots in the economics literature. Under these assumptions, we give an algorithm that for any dimension $d > 2$ guarantees that all downstream agents with at most $k$ actions have swap regret diminishing at a rate of $\Tilde{O}(T^{2/3})$  (Theorem \ref{thm:quantalresponse}), with only a polynomial dependence on the dimension $d$. Under a perhaps less realistic behavioral assumption (that agents play best responses with respect to a discretization of their utility function), we show how to guarantee all downstream agents regret at the optimal rate of $\Tilde{O}(\sqrt{T})$ for every dimension $d$ (Theorem \ref{thm:snappedutilities}). This stands in sharp contrast to calibration bounds, which have a dimension dependence in their exponent. Our focus here is on achievable rates; as with calibration, the computational complexity of our algorithms scales exponentially with $d$, and achieving similar rates in high dimensional settings with computationally efficient algorithms is in our view the most compelling open problem coming out of this work. 

\subsection{Related Work}
The idea of sequential calibration dates to \citet{dawid1982well}, and \citet{foster1998asymptotic} were the first to show that it is possible to maintain calibrated forecasts in a sequential adversarial setting. \citet{foster1997calibrated} connected calibration to decision making in a game theoretic context, and showed that agents who best respond to calibrated forecasts of the utilities of their actions obtain diminishing internal (equivalently swap) regret and hence converge to correlated equilibrium. The correct ``rates'' at which calibration error can be driven to $0$ in an online setting has been an important open question; the best known rates in $1$ dimension are $T^{2/3}$ and degrade exponentially with the dimension \citep{okoroafor2023faster}. In $1$ dimension, it is known that it is impossible to obtain $\sqrt{T}$ rates --- \citet{qiao2021stronger} show a lower bound of $\Omega(T^{0.528})$. This is despite the fact that it is known that swap regret for a single agent can be obtained at a rate of $O(\sqrt{T})$ \citep{blum2007internal}. For a more detailed overview on calibration, we refer the reader to Appendix \ref{app:calibration}.

\citet{kakade2008deterministic} and \citet{foster2018smooth} introduce \emph{smooth} calibration, and give --- as we do --- guarantees for agents who smoothly best respond. The main benefit of smooth calibration is that it can be guaranteed with deterministic algorithms in adversarial environments (unlike calibration, which requires randomization) --- and that players who smoothly best respond to smoothly calibrated forecasts converge to Nash (rather than correlated) equilibrium. The algorithms for smooth calibration do not converge at rates faster than the best known rates for calibration, however. 

The most closely related work is \citet{kleinberg2023ucalibration}, who take the perspective of a forecaster whose goal is to guarantee low external regret for arbitrary downstream agents. They define ``U-calibration" as a measure of forecast quality that is necessary and sufficient to achieve this goal and provide algorithms that minimize U-calibration error. In comparison, our work gives stronger guarantees than low external regret, and we consider general $d$-dimensional prediction tasks, while their work focuses on 1-dimensional and multiclass predictions (i.e. predictions that are distributions over $d$ outcomes, which is a special case of the setting we consider). Subsequently to the initial publication of our paper, \citet{hu2024predict} extend our result to remove the dependence on the number of actions $k$ in the special $d=1$ dimensional case that is focal in \citet{kleinberg2023ucalibration}. In particular, their result shows that for $d=1$, it is possible to obtain diminishing \emph{swap} regret at the same rates at which \citet{kleinberg2023ucalibration} obtain diminishing external regret. We discuss in Appendix \ref{app:calibration} the relationship between swap regret and different measures of calibration error, which is what resolves a seeming contradiction in the results of \citet{kleinberg2023ucalibration} and the results of this paper and \citet{hu2024predict}.

This work is conceptually related to a recent line of work on \emph{decision calibration} \citep{zhao2021calibrating} and  \emph{omniprediction} that has the goal of learning a single predictor that is a sufficient statistic to minimize any downstream loss function in some family, with respect to some benchmark family of functions \citep{gopalan2022omnipredictors,gopalan2023loss,hu2023omnipredictors,globus2023multicalibrated}. Most of this literature focuses on the batch setting, but \citet{garg2024oracle} give omnipredictors in the online adversarial setting. The omniprediction literature is focused on $1$-dimensional binary-valued prediction, whereas we are interested in high-dimensional, real-valued prediction. \citet{noarov2023highdimensional} consider making high-dimensional predictions that guarantee diminishing swap regret to downstream agents. We take direct inspiration from their work and make essential use of their forecasting algorithm. The main distinction is that \citet{noarov2023highdimensional} tailors their predictions to the best response functions of particular downstream agents, whereas our goal is to give guarantees that hold simultaneously for \emph{all} downstream agents.

In our main result for high-dimensional settings, we require agents to smoothly best respond, i.e. to choose actions randomly with probability proportional to their utility. This is consistent with the literature on smooth calibration \citep{kakade2008deterministic,foster2018smooth}, but one may ask: is this natural behavior? To offer some justification, we note that our formulation of smooth best response fits into a broader theory of \textit{quantal choice} that studies probabilistic best responses, otherwise known as \textit{quantal responses} \citep{mcfadden1976quantalchoice, mckelvey1995quantalresponse}. The quantal choice model posits that agents are not perfectly rational; instead, they make slight errors in choosing their actions.

\section{Model and Preliminaries}

\subsection{Predictions and Decisions}

We consider a repeated interaction between a learner, an adversary, and downstream agents. We use $\cC \subseteq \R^d$ to denote a convex prediction space for a finite dimension $d$. Without loss of generality, we consider $\cC \subseteq [0,1]^d$ (rescaling the predictions space degrades our bounds by only a multiplicative constant). In rounds $t \in [T]$, the learner interacts with an adversary as follows: 
\begin{enumerate}
    \item The adversary chooses a distribution over outcomes $Y_t \in \Delta \cC$.
    \item The learner produces a distribution over predictions $p_t \in \Delta \cC$, from which a prediction $\hat{y}_t$ is sampled.
    \item The adversary reveals an outcome $y_t \sim Y_t$. 
\end{enumerate}
This interaction accumulates in a transcript $\pi_T$ of predictions $\hat{y}_1,...,\hat{y}_{t-1}$ and outcomes $y_1,...,y_T$. The learner's output is a (possibly randomized) function of previous predictions $\hat{y}_1,...,\hat{y}_{t-1}$ and outcomes $y_1,...,y_{t-1}$. In this setting, the adversary's choice of outcome can depend on previous predictions $\hat{y}_1,...,\hat{y}_{t-1}$ and outcomes $y_1,...,y_{t-1}$ but cannot depend on the learner's realization of randomness on round $t$.

We want to study the value of predictions for downstream agents. Upon observing a prediction $\hat{y}_t$, an agent chooses an action $a_t$ from a finite action set $\cA$. An agent's utility depends on the chosen action and the outcome. More specifically, when the outcome is $y$, the agent receives utility $u(a_t, y)$ according to a bounded utility function $u: \cA \times \cC \to [0,1]$. We will assume that utility functions are affine and Lipschitz-continuous in $y$. 

\begin{assumption}
    Fix a utility function $u: \cA \times \cC \to [0,1]$. We assume that for every action $a \in \cA$, $u(a, y)$ is affine in $y$, i.e. $u(a,y) = \<v_a, y\> + c_a$ for some $v_a \in \R^d, c_a \in \R$. Moreover, we assume that $u(a,y)$ is $L$-Lipschitz in $y$ in the $\ell_\infty$ norm: for any $y_1,y_2 \in \cC$, $|u(a, y_1) - u(a, y_2)| \leq L\|y_1 - y_2\|_\infty$.
\end{assumption}

\begin{remark}
    Note that by linearity of expectation this is only more general than considering expectation maximizing agents that have arbitrary utility functions over $d$ discrete outcomes, and states/forecasts $y \in \Delta [d]$ that correspond to probability distributions over these $d$ outcomes. Thus, our setting generalizes the settings of e.g. \citet{zhao2021calibrating,kleinberg2023ucalibration}.
\end{remark}

We will later on take advantage of the fact that it is not more restrictive to consider utility functions that are \textit{linear} in $y$. 
\begin{observation}
    Any utility function over a $d$-dimensional prediction space that is affine in $y$ can be equivalently expressed as a utility function over a $(d+1)$-dimensional prediction space that is linear in $y$. We simply set the $(d+1)^{st}$ coordinate of the outcome/prediction to always take value 1. This preserves the convexity of the outcome/prediction space.
\end{observation}

We measure agent regret by comparing the utility the agents obtain from the sequence of chosen actions against the counterfactual utility they would have received under another sequence of actions. We focus on swap regret, where the counterfactual action sequence is obtained by applying an action modification rule to the realized sequence of chosen actions. 

\begin{definition}[Swap Regret]
    Fix an agent with action set $\cA$ and utility function $u$. For a sequence of actions $a_1,...,a_T \in \cA$ and outcomes $y_1,...,y_T$, the agent's swap regret to an action modification rule $\phi: \cA \to \cA$ is: 
    \[
    Reg(u, \phi) = \frac{1}{T} \sum_{t=1}^T \left( u(\phi(a_t), y_t) - u(a_t, y_t) \right)
    \]
    More generally, when an agent plays distributions over actions $q_1,...,q_T \in \Delta\cA$, we measure swap regret to $\phi$ by:
    \[
    Reg(u, \phi) = \frac{1}{T} \sum_{t=1}^T \E_{a\sim q_t}[ u(\phi(a), y_t) - u(a, y_t) ]
    \]
    We say that an agent has swap regret $\alpha$ if for any action modification rule $\phi$, $Reg(u, \phi) \leq \alpha$. 
\end{definition}

The objective, from the learner's perspective, is to produce a sequence of predictions such that any agent responding to these predictions as if they were correct experiences low swap regret, regardless of their utility function.

\subsection{Conditionally Unbiased Predictions}
While we ultimately evaluate predictions based on the utility obtained by downstream agents who act according to our predictions, we will ask for predictions to satisfy some intermediate properties along the way---in particular, unbiasedness subject to a collection of events. Recall that  calibration can be summarized as unbiasedness conditional on events corresponding to the predictions themselves. We will make use of a relaxation which asks for unbiasedness subject to a more coarsely defined collection of events. Our algorithms will derive from the general framework developed by \citet{noarov2023highdimensional}. For the remainder of this section, we introduce the key tools we use. First, we formalize the notion of events and conditional bias.

\begin{definition}[Conditional Bias on Events]
    Fix a transcript $\pi_T$. Let $\cE$ be a collection of events, where each event is defined by a subsequence indicator function $E: \cC \to \{0, 1\}$. Let $n_T(E) = \sum_{t=1}^T \E_{\hat{y}_t\sim p_t}[E(\hat{y}_t)]$ denote the (expected) frequency of event $E$ up to round $T$. We say that $\pi_T$ has bias $\alpha$ conditional on $\cE$ for some $\alpha:\mathbb{R}\rightarrow \mathbb{R}$ if for all $E \in \cE$:
    \[
    \frac{1}{T} \left\| \sum_{t=1}^T \E_{\hat{y_t}\sim p_t}[E(\hat{y_t})(\hat{y_t} - y_t)] \right\|_\infty \leq \alpha(n_T(E))
    \]
\end{definition}

\citet{noarov2023highdimensional} show how to achieve unbiased predictions at favorable rates conditional on any collection of events. The algorithm requires the learner to randomize over a finite prediction space (we defer the details to \citet{noarov2023highdimensional}). We next define a discretized prediction space. 

\begin{definition}[$\eps$-net]
    We say $\cC_\eps \subset \cC$ is an $\eps$-net of $\cC$ in the $\ell_\infty$ norm if for every $y \in \cC$, there exists $y_\eps \in \cC_\eps$ such that $\|y - y_\eps\|_\infty \leq \eps$. Observe that we can always obtain an $\eps$-net of size $|\cC_\eps| \leq \left(\frac{1}{\eps}\right)^d$ by simply discretizing each coordinate into multiples of $\eps$. 
\end{definition}

When the learner produces predictions from a discretized prediction space $\cC_\eps$, \citet{noarov2023highdimensional} show that we can make predictions with low bias on any collection of events, with bias scaling logarithmically in the number of events. In our results, we will make direct use of the algorithm given by \citet{noarov2023highdimensional}. For reference, the algorithm is presented in Appendix \ref{app:alg}.

\begin{theorem} \label{thm:biasbound} \citep{noarov2023highdimensional}
    For a collection of events $\cE$ and convex prediction/outcome space $\cC$, there is an algorithm producing predictions $p_1,...,p_T \in \Delta \cC_\eps$ such that for any sequence of outcomes $y_1,...,y_T \in \cC$ chosen by the adversary, the  bias conditional on $\cE$ is at most:
    \[
    \alpha(n_T(E)) \leq O\left( \frac{\ln(d|\cE|T) + \sqrt{\ln(d|\cE|T) \cdot n_T(E) }}{T} + \eps \right) \leq O\left( \sqrt{\frac{\ln(d|\cE|T)}{T}} + \eps \right)
    \]
    For $\eps \leq 1/\sqrt{T}$, the bias is at most $O\left( \sqrt{\frac{\ln(d|\cE|T)}{T}} \right)$.
    The algorithm can be implemented with per-round running time scaling polynomially in $d$, $|\cE|$, and $T$.
\end{theorem}

\section{Predictions in Low Dimensions}\label{sec:low}

In this section we show how to use the framework of unbiased prediction to produce predictions for all possible agents in the one- and two-dimensional case (Theorem \ref{thm:1d} and Theorem \ref{thm:2d} respectively). To accomplish this, we build on a key result of \citet{noarov2023highdimensional} connecting conditionally unbiased predictions to the swap regret of downstream agents. For a particular instantiation of events---concretely, the \textit{best response correspondences} of agents---conditional bias has a natural implication for  swap regret. Next we define an agent's best response and their corresponding best response events.

\begin{definition}[Best Response]
    Consider an agent with utility function $u$ over an action set $\cA$ and prediction space $\cC$. The agent's best response to $y \in \cC$ according to their utility function $u$ is the action $a = BR(u, y)$ where $BR(u, y) = \argmax_{a\in\cA} u(a, y)$. 
\end{definition}

\begin{definition}[Best Response Events]
    Fix an action set $\cA$ and utility function $u$. For any action $a \in \cA$, we define its corresponding best response event:
    \[
    E_{u, a}(y) = \1[BR(u, y) = a]
    \]
\end{definition}

A sufficient condition for minimizing swap regret for any agent is minimizing bias conditional on their best response events. More accurately, an agent with action set $\cA$ and utility function $u$ can achieve low swap regret when they best respond to predictions that are unbiased subject to the collection of events $\cE_u = \{E_{u,a}\}_{a\in\cA}$. Informally, this unbiasedness condition ensures that for every action $a$, agents' estimates of their utilities are (on average) correct whenever they choose to take action $a$. Therefore, since agents pick their utility-maximizing action, they must be choosing the action that is best (on average) conditional on that choice --- and thus not have swap regret. We state the resulting swap regret bound in the next theorem but defer further details of the argument to \citet{noarov2023highdimensional}. 

\begin{theorem} \label{thm:swapregretbound} \citep{noarov2023highdimensional}
    Fix a transcript $\pi_T$. Consider an agent with action set $\cA$ and utility function $u$ who, at every round $t\in[T]$, takes action $a_t = BR(u, \hat{y}_t)$, where $\hat{y}_t \sim p_t$. Let $E_{u,a}$ be the best response event corresponding to $u$ and $a \in \cA$. Then, if $\pi_T$ has bias $\alpha$ conditional on events $\cE$ such that $\{E_{u,a}\}_{a\in\cA} \subseteq \cE$, the agent has swap regret bounded by:
    \[
    \max_{\phi:\cA\to\cA} Reg(u, \phi) \leq 2L\sum_{a\in\cA} \alpha(n_T(E_{u,a}))
    \]
    For any concave $\alpha$, this is at most: $\max_{\phi:\cA\to\cA} Reg(u, \phi) \leq 2L|\cA| \alpha(T/|\cA|)$. In particular, if we plug in the bound obtained by Theorem \ref{thm:biasbound}, then we have: 
    \[
    \max_{\phi:\cA\to\cA} Reg(u, \phi) \leq O\left( L \sqrt{\frac{|\cA|\ln(d|\cE|T)}{T}} \right)
    \]
\end{theorem}

This implies that given any set of utility functions $\cU$, running the algorithm of Theorem \ref{thm:biasbound} with the corresponding collection of best response events $\{E_{u,a} : u\in\cU, a\in\cA\}$ promises any agent regret at most $O\left( L \sqrt{|\cA| \ln(d|\cU||\cA|T)/T} \right)$, as long as their utility function belongs to $\cU$. 

Our present goal is to promise low swap regret with respect to \textit{any} utility function an agent may possess. We will approach this by considering the augmented collection of best response events corresponding to not just any particular set $\cU$, but the set of \textit{all possible} utility functions. If our predictions are unbiased subject to all possible best response events, then we can be certain that our predictions are unbiased subject to the best response events of any particular agent. However, recall that the bounds on conditional bias given by Theorem \ref{thm:biasbound} scale with the number of events. Thus, we must be careful to ensure that we define a reasonably sized collection of events. We show that for one- and two-dimensional prediction spaces, the collection of best response events is not too large and hence we can derive meaningful swap regret bounds.  

To do so, we will take advantage of a key structural property of the best response function---namely, that its level sets are convex.

\begin{lemma} \label{lem:convexregions}
    For any utility function $u$ that is affine in $y$ , the level set of the corresponding best response function $\{y : BR(u, y) = a\}$ is convex for all $a \in \cA$, i.e. for any $y_1, y_2 \in \cC$ and any $\lambda \in [0,1]$, if $BR(u, y_1) = BR(u, y_2) = a$, then $BR(u, \lambda y_1 + (1-\lambda) y_2) = a$. 
\end{lemma}
\begin{proof}
    We use the definition of an affine function to compute, for any alternate action $a' \in \cA$:
    \begin{align*}
        u(a, \lambda y_1 + (1-\lambda) y_2) &= \<v_a, \lambda y_1 + (1-\lambda) y_2\> + c_a \\
        &= \lambda \<v_a, y_1\> + (1-\lambda) \<v_a, y_2\> + c_a \\
        &= \lambda u(a, y_1) + (1-\lambda) u(a, y_2) + c_a - \lambda c_a - (1-\lambda) c_a \\
        &= \lambda u(a, y_1) + (1-\lambda) u(a, y_2) \\
        &\geq \lambda u(a', y_1) + (1-\lambda) u(a', y_2) \\
        &= u(a', \lambda y_1 + (1-\lambda) y_2)
    \end{align*}
    where the inequality follows by definition of the best response function and the last line follows from applying the same transformation. Hence, $a = BR(u, \lambda y_1 + (1-\lambda) y_2)$, proving the claim. 
\end{proof}

In other words, the set of predictions that activate a best response event form a convex set. As a result, we can enumerate all possible best response events (across all possible utility functions) by enumerating all \textit{convex hulls} of predictions in our prediction space. Our results in this section proceed by bounding the number of such convex hulls when the prediction space is an $\eps$-net of $\cC$ (as required by the algorithm of Theorem \ref{thm:biasbound}). Given a bound on the number of such convex hulls,  an application of Theorem \ref{thm:swapregretbound} bounds the swap regret for any downstream agent. 

The $d=1$ case is particularly elegant: convex hulls over any finite set of grid points on the unit interval $\cC_\eps \subseteq [0,1]$ are simply intervals with endpoints $y_1,y_2 \in \cC_\eps$, i.e. they can be specified by two points in $\cC_\eps$. Hence, the number of convex hulls---and thus, the number of events we must define---is only quadratic in the discretization parameter $\eps$. We show that setting $\eps$ to optimally trade off between the number of events and the precision of predictions, we can guarantee the optimal  swap regret rate of $O(\sqrt{T})$ (up to logarithmic factors) simultaneously for every downstream agent.

\begin{theorem}\label{thm:1d}
    Consider the prediction/outcome space $\cC = [0,1]$ and let $\cC_\eps$ be an $\eps$-net of $\cC$. Let $\cE = \{ E_{y_1,y_2} : y_1, y_2 \in \cC_\eps \}$ be the collection of events $E_{y_1,y_2}(y) = \1[y \in [y_1, y_2]]$. Then, using the algorithm given by \citet{noarov2023highdimensional} and setting $\eps=1/\sqrt{T}$, an agent who best responds according to action set $\cA$ and utility function $u$  obtains swap regret bounded by:
    \[
    \max_{\phi: \cA \to \cA} Reg(u, \phi) \leq O\left( L|\cA| \sqrt{ \frac{\ln T}{T} } \right)
    \]
Moreover, the forecasts can be made with per-round running time polynomial in $T$. 
\end{theorem}
\begin{proof}
    As we remarked above, convex sets in one dimension are simply intervals. Therefore, given that our predictions lie on a finite $\eps$-net, it follows from Lemma \ref{lem:convexregions} that for any utility function $u$ and any action $a$, the best response event $E_{u,a}$ exactly corresponds to an event $E_{y_1,y_2} \in \cE$. That is, for any action $a \in \cA$, there exist $y_1,y_2 \in \cC_\eps$ such that for all $y \in \cC_\eps$, $BR(u, y) = a$ if and only if $E_{y_1,y_2}(y) = 1$. Thus to prove the regret bound, all we need is to determine $|\cE|$. Since $|\cC_\eps| \leq \frac{1}{\eps}$, we have that $|\cE| \leq \frac{1}{\eps^2}$. Plugging this into Theorem \ref{thm:biasbound} and using our setting of $\eps$, we have that the bias conditional on every event $E_{y_1,y_2} \in \cE$ is bounded by:
    \[
    \alpha(n_T(E_{y_1,y_2}) \leq O\left( \sqrt{ \frac{\ln(T/\eps^2)}{T} } + \eps \right) = O\left( \sqrt{ \frac{\ln T}{T} } \right)
    \]
 Thus by Theorem \ref{thm:swapregretbound}, any agent with utility function $u$ and action set $\cA$ has swap regret at most:
    \[
    \max_{\phi: \cA \to \cA} Reg(u, \phi) \leq O\left( L|\cA| \sqrt{ \frac{\ln T}{T} } \right)
    \]
\end{proof}

For $d = 1$, the set of all convex hulls over a discrete prediction space was particularly simple: the set of intervals over the same discretization. In higher dimensions, characterizing the set of all convex hulls of discrete points is considerably more difficult. The main ingredient for our result in the $d=2$ dimensional case is the following theorem shown by \citet{ivic1994digital}, which establishes the number of convex hulls of points in a two-dimensional grid space. Thus, taking our predictions to lie in such a grid space, the number of possible best response events directly follows from this result.

\begin{theorem}\label{fact:digitalpolygons}\citep{ivic1994digital}
    Consider an integer grid of size $m \times m$. Let $D(m)$ denote the number of convex polygons whose vertices have integer coordinates (i.e. number of convex hulls of grid points) that are distinct up to translation. Then, there exists a constant $C$ such that $\ln D(m) \leq C m^{2/3}$. 
\end{theorem}

\begin{remark}\label{rem:digitalpolygons}
    To count the total number of convex polygons $I(m)$ with integer vertices in the $m\times m$ integer grid, observe that there are at most $m^2$ possible translations of any polygon in the grid to another polygon in the grid with integer vertices. Thus, $I(m) \leq m^2 D(m)$. 
\end{remark}

\begin{theorem} \label{thm:2d}
    Consider a prediction space $\cC = [0,1]^2$ and let $\cC_\eps$ be an $\eps$-net of $\cC$ obtained by discretizing each coordinate to multiples of $\eps$. Let $\cP_\eps$ be the set of all convex polygons with vertices belonging to $\cC_\eps$. Define $E_{P}(y) = \1[y \in P]$ to be the event indicating whether $y$ lies in a polygon $P$. Then, using the algorithm of \citet{noarov2023highdimensional} with $\cE = \{ E_{P} \}_{P\in\cP_\eps}$ and $\eps = \frac{1}{T^{3/8}}$, an agent with any action set $\cA$ and utility function $u$ who best responds to the forecasts obtains swap regret bounded by:
    \[
    \max_{\phi: \cA \to \cA} Reg(u, \phi) \leq O\left( \frac{L|\cA| \sqrt{\ln(T)}}{T^{3/8}} \right)
    \]
\end{theorem}
\begin{proof}
    By construction of $\cC_\eps$, our predictions lie in a $\frac{1}{\eps} \times \frac{1}{\eps}$ grid, and so $|\cE| = I(\frac{1}{\eps})$. Plugging into Theorem \ref{thm:biasbound}, it follows from Theorem \ref{fact:digitalpolygons} and Remark \ref{rem:digitalpolygons} that the bias conditional on every event $E_P \in \cE$ is at most:
    \begin{align*}
        \alpha(n_T(E_P)) \leq O\left( \sqrt{\frac{\ln(I(\frac{1}{\eps})T)}{T}} + \eps \right) \leq O\left( \sqrt{\frac{\ln(\frac{T}{\eps^2} D(\frac{1}{\eps}))}{T}} + \eps \right) \leq O\left( \sqrt{\frac{\frac{1}{\eps^{2/3}}\ln(\frac{T}{\eps^2})}{T}} + \eps \right)
    \end{align*}
    For $\eps = \frac{1}{T^{3/8}}$, we get $\alpha(n_T(E_P)) \leq O\left( \frac{\sqrt{\ln(T)}}{T^{3/8}} \right)$. By Theorem \ref{thm:swapregretbound}, we can conclude that the agent has swap regret at most:
    \[
    \max_{\phi: \cA \to \cA} Reg(u, \phi) \leq O\left( \frac{L|\cA| \sqrt{\ln(T)}}{T^{3/8}} \right)
    \]
\end{proof}

\section{Predictions in Higher Dimensions} \label{sec:high}

In the previous section we showed that by taking advantage of nice structural properties of best response correspondences, we can---in one- and two-dimensional settings---straightforwardly enumerate the set of all best response events across all possible downstream utility functions.  Making predictions that are unbiased with respect to this set yields low swap regret for all possible downstream agents who best respond to forecasts. Turning now to the general setting, i.e. where we consider $d$-dimensional predictions, this approach will no longer be tractable. This is because our previous results crucially rely on the fact that in low dimensions, convex regions over grid points---and correspondingly best response events induced by predictions made over grid points---have \textit{boundedly low complexity}. Most saliently, convex regions in one dimension are specified by just two points. In two dimensions, already, convex regions on grid points are substantially more complex ---although fortunately, we can exploit a geometric result bounding the number of convex regions. In higher dimensions, we cannot in general expect to furnish a reasonably sized collection of events in this way. Naively, the number of convex hulls taken over $\left(\frac{1}{\eps}\right)^d$ discrete predictions is doubly exponential in $d$ (remember that we can only obtain bias that depends logarithmically in the number of events). 

Thus to solve the general case, we take a different approach: we enumerate all possible utility functions---up to some discretization, and then subsequently define events corresponding to agent responses for agents with utility functions within this discretization. Recall that the bias obtained by Theorem \ref{thm:biasbound} depends on the number of conditioning events. The space of utility functions is continuously large, and so in order to obtain reasonable bounds, we cannot discretize the utility functions too finely.

When we restrict our attention to this finite net of utility functions, we can only hope to approximate the utility function of any particular agent. This is problematic due to the fact that the best response function is \textit{discontinuous}; best response events corresponding to this discretized set of utility functions are not guaranteed to be close to  an agent's actual best responses, even if their utility function is close to that of one of the discretized utility functions. We illustrate this in the following lemma, which states that there are utility functions such that even as the difference between their payoffs approach 0, bias conditional on their best response events differ by a constant, and the cumulative swap regret of best responding agents differ by $\Omega(T)$. 

\begin{lemma}
    Fix $\delta \in (0,0.5)$. Consider a prediction and outcome space $\cC = [0,1]$. There exist utility functions $u, \Tilde{u}$ and a sequence of predictions $\hat{y}_1,...,\hat{y}_T$ and outcomes $y_1,...,y_T$ such that $|u(a, y) - \Tilde{u}(a, y)| \leq \frac{2\delta}{1+2\delta}$ for all $a\in\cA, y\in\cC$, but $\alpha(E_{u,a}) = 0$ while $\alpha(E_{\Tilde{u},a}) = 0.5 + \delta$, where $E_{u,a}$ and $E_{\Tilde{u},a}$ denote the best response events corresponding to $u$ and $\Tilde{u}$ respectively. Furthermore, an agent who takes action $a_t = BR(\Tilde{u}, \hat{y}_t)$ for all $t \in [T]$ accumulates average swap regret $\max_{\phi:\cA\to\cA} Reg(\Tilde{u},\phi) = 1$.
\end{lemma}
\begin{proof}
    Consider an agent who chooses from two actions $\cA = \{0,1\}$. We will define $u$ so that the best response function induced by $u$ is a thresholding at $y = 0.5 - \delta$: an agent chooses action 1 if and only if $y \geq 0.5 - \delta$ (we assume tie-breaking in favor of action 1). Similarly, $\Tilde{u}$ will induce a thresholding at $y = 0.5$. Formally, we can write $u(a, y) = \frac{1}{1+2\delta}(a(y+\delta) + (1-a)(1-y-\delta) + \delta)$ and $\Tilde{u}(a, y) = ay + (1-a)(1-y)$. It is easy to check that both utility functions are affine in $y$. Moreover, we can compute their absolute difference: for $a=1$, we have that
        \begin{align*}
            |u(1, y) - \Tilde{u}(1, y)| = \left|\frac{y+2\delta}{1+2\delta} - y \right| = \left| \frac{2\delta}{1+2\delta} - y\left(\frac{2\delta}{1+2\delta}\right) \right| \leq \frac{2\delta}{1+2\delta}
        \end{align*}
        \item and for $a=0$, we have that
        \begin{align*}
            |u(0, y) - \Tilde{u}(0, y)| = \left|\frac{1-y}{1+2\delta} - (1-y) \right| = \left| y\left(\frac{2\delta}{1+2\delta}\right) - \frac{2\delta}{1+2\delta} \right| \leq \frac{2\delta}{1+2\delta}
        \end{align*}
    Now, consider a sequence of predictions where $\hat{y}_t = 0.5 - \delta$ for odd $t$ and $\hat{y}_t = 0.5 + \delta$ for even $t$. For odd $t$, the outcome is $y_t = 1$, and for even $t$, the outcome is $y_t = 0$. 
    
    An agent who best responds according to $u$ takes action $a_t = 1$ at every round. Thus, the bias conditioned on the best response event corresponding to $u$ is the bias taken over the entire sequence, and we can see that $\alpha(E_{u,1}) = \left|\frac{1}{T}\sum_{t=1}^T \1[BR(u, \hat{y}_t) = 1] (\hat{y}_t - y_t)\right| = 0$. On the other hand, an agent who best responds according to $\Tilde{u}$ takes action $a_t = 0$ on odd rounds and action $a_t = 1$ on even rounds. In this case, observe that the bias conditioned on the best response events corresponding to $\Tilde{u}$ is: $\alpha(E_{\Tilde{u},0}) = \alpha(E_{\Tilde{u},1}) = 0.5 + \delta$. Furthermore, $\Tilde{u}(a_t, y_t) = 0$ for all $t$, while applying a counterfactual swap would grant the agent utility 1 at every round. The claim then follows. 
\end{proof}

To circumvent this issue, we consider agents who play a \textit{continuous} approximation to their best response, which will, intuitively speaking, \textit{smooth out} the discontinuities of strictly best responding. As we will discuss, this modelling of agent behavior is well-founded in an extensive line of work on quantal choice in the economics literature. Our main result in this section is Theorem \ref{thm:quantalresponse}, which states that there is an algorithm that guarantees diminishing swap regret at a rate of $\tilde O(T^{2/3})$ for any possible agent with $k$ actions who \textit{smoothly} best responds to our forecasts, in any dimension $d$. We remark that while this result substantially improves upon what is achievable via calibrated predictions, we do not recover optimal swap regret bounds of $O(\sqrt{T})$. Along the way, however, we show (in Theorem \ref{thm:snappedutilities}) how to recover $O(\sqrt{T})$  swap regret for agents who respond in a slightly different, but perhaps less realistic, manner --- in particular, agents who act by best responding according not to their own utility function, but to the closest utility function in an appropriately defined $\delta$-cover of utility functions. We remark that, unlike in our previous results, our algorithms in this section require as input a parameter $k$ that serves as an upper bound on the number of actions available to any downstream agent to whom we give guarantees.

\subsection{Discretizing the Space of Utility Functions}

We begin by describing a suitable discretization of utility functions. Recall that to linearize a utility function that is affine over a $d$-dimensional prediction space, we can simply augment the prediction space by one dimension and consider instead the $(d+1)$-dimensional prediction space $\cC \subseteq [0,1]^{d+1}$ where every $y \in \cC$ takes value 1 in its last coordinate. We will overload notation by extending affine utility functions over a $d$-dimensional prediction space to linear functions over a $(d+1)$-dimensional prediction space in this way and henceforth restrict our attention to utility functions that are \textit{linear} in $y$. 

We will choose a discretization that approximates any possible utility function within a distance of $\delta$.

\begin{definition}[$\delta$-Cover of Utility Functions]
    Let $\cU$ be a set of utility functions. We say that $\cU_\delta$ is a $\delta$-cover of $\cU$ if for every $u \in \cU$, there exists $u_\delta \in \cU_\delta$ such that for every $a \in \cA$ and $y \in \cC$, $|u(a, y) - u_\delta(a, y)| \leq \delta$.
\end{definition}

\begin{lemma}\label{lem:cover}
    Let $\cU^k$ be the set of utility functions over an action set $\cA$ of size at most $k$ and a prediction space $\cC \subseteq [0,1]^{d+1}$ that are linear in their second argument. For some parameter $\delta > 0$, let $\delta' = \delta(d+1)$. Then, there is a $\delta'$-cover of $\cU^k$, denoted by $\cU^k_{\delta'}$, with $|\cU^k_{\delta'}| \leq \left(\frac{1}{\delta}\right)^{k(d+1)}$. 
\end{lemma}
\begin{proof}
    Consider any $u \in \cU^k$. By linearity of $u$ in its second argument, for any $i\in\cA$ we can write $u(i,y) = \<v^i, y\>$ from some $v^i \in [0,1]^{d+1}$ (note that entries of $v^i$ must be bounded within $[0,1]$, since $u(i,y)$ is bounded within $[0,1]$). In other words, since $|\cA| \leq k$, $u$ is parameterized by a set of at most $k$ vectors $v^1,...,v^{k} \in [0,1]^{d+1}$. 
    
    We construct a $\delta'$-cover of $\cU^k$ by considering a discretized set of representative vectors. Specifically, let $\cV_\delta \subseteq [0,1]^{d+1}$ be a $\delta$-net of $[0,1]^{d+1}$---that is, for every $v \in [0,1]^{d+1}$, there exists $v_\delta \in \cV_\delta$ such that $\|v - v_\delta\|_\infty \leq \delta$. Observe that we can obtain such a $\delta$-net by discretizing each coordinate to multiplies of $\delta$. Thus, we have $|\cV_\delta| \leq \left(\frac{1}{\delta}\right)^{d+1}$. Define $\cU_{\delta'}^k$ to be the set of utility functions parameterized by discretized vectors $v^1_\delta,...,v^k_\delta \in \cV_\delta$. That is, given any $u_{\delta'} \in \cU_{\delta'}^k$, we have for every $i\in\cA$, $u_{\delta'}(i, y) = \<v^i_\delta, y\>$ for some $v^i_\delta \in \cV_\delta$. 
    
    By the fact that $\cV_\delta$ is a $\delta$-net, it follows that for any $u \in \cU^k$ parameterized by $v^1,...,v^k$, we can find a nearby $u_{\delta'} \in \cU_{\delta'}^k$ parameterized by $v^1_\delta,...,v^k_\delta$ such that $\| v^i - v^i_\delta \|_\infty \leq \delta$ for all $i \in [k]$. Thus, we have that $\cU_{\delta'}^k$ is indeed a $\delta'$-cover of $\cU^k$:
    \[
    |u(i, y) - u_{\delta'}(i, y)| = |\< v^i - v^i_\delta, y \>| \leq \|v^i - v^i_\delta\|_\infty \cdot (d+1) \leq \delta (d+1)
    \]
    Moreover, we have that $|\cU_{\delta'}^k| = |\cV_\delta|^k \leq \left(\frac{1}{\delta}\right)^{k(d+1)}$, which completes the proof. 
\end{proof}

With this discretized set of utility functions in hand, our next result shows how to recover $\Tilde{O}(\sqrt{T})$ swap regret bounds under a particular assumption on agent behavior. Imagine an agent with utility function $u$ who computes their best response not according to $u$, but rather according to the utility function given by ``snapping" $u$ to a nearby $u_{\delta'}$ in the $\delta'$-cover. By definition of a $\delta'$-cover, $u_{\delta'}$ closely approximates $u$, and hence, the agent does not lose too much utility---at most $\delta'$ each round---by best responding to $u_{\delta'}$ instead of $u$. Then, as we shall see, requiring our predictions to be unbiased conditional on the best response correspondences of the $\delta'$-cover yields optimal swap regret guarantees (up to logarithmic factors) for the agent---of course as measured with respect to their real utility function $u$. 

\begin{theorem} \label{thm:snappedutilities}
    Fix a transcript $\pi_T$. Let $\cU^k$ be the set of utility functions over an action set of size at most $k$ and prediction space $\cC$ that are linear in its second argument. Let $\cU_{\delta'}^k$ be the $\delta'$-cover of $\cU^k$ given by Lemma \ref{lem:cover}, where $\delta' = \delta(d+1)$ for some parameter $\delta > 0$. Define $\cE_{\cU_{\delta'}^k}$ to be the collection of best response events corresponding to all $u_{\delta'} \in \cU_{\delta'}^k$:
    \[
    \cE_{\cU_{\delta'}^k} = \{ E_{u_{\delta'}, a}(y) = \1[BR(u_{\delta'}, y) = a] : u_{\delta'} \in \cU_{\delta'}^k, a \in [k] \}
    \]
    Consider an agent with action set $\cA$ such that $|\cA|\leq k$, and utility function $u \in \cU^k$. Suppose at every round $t \in [T]$, the agent takes action $a_t = BR(u_{\delta'}, \hat{y}_t)$ for a nearby $u_{\delta'}$ satisfying $|u(a, y) - u_{\delta'}(a, y)| \leq \delta'$ for all $a \in \cA, y \in \cC$. Then, setting $\delta = \frac{1}{(d+1)\sqrt{T}}$ and using the prediction algorithm of \citet{noarov2023highdimensional} with the set of events $\cE_{\cU_{\delta'}^k}$, the agent has swap regret bounded by:
    \[
    \max_{\phi: \cA \to \cA} Reg(u, \phi) \leq O\left( L \sqrt{\frac{|\cA|dk\ln(Tdk)}{T}} \right)
    \]
\end{theorem}
\begin{proof}
    By Lemma \ref{lem:cover}, we have that $|\cU_{\delta'}^k| \leq \left(\frac{1}{\delta}\right)^{k(d+1)}$. Thus $|\cE_{\cU_{\delta'}^k}| \leq k \left(\frac{1}{\delta}\right)^{k(d+1)}$. Plugging this into Theorem \ref{thm:biasbound}, we have that for sufficiently small $\eps$, the bias conditional on $\cE_{\cU_{\delta'}^k}$ is at most:
    \[
    \alpha(n_T(E_{u_{\delta'}, a})) \leq O\left( \frac{\ln(Tdk\left(\frac{1}{\delta}\right)^{k(d+1)}) + \sqrt{\ln(Tdk\left(\frac{1}{\delta}\right)^{k(d+1)}) \cdot n_T(E_{u_{\delta'}, a}) }}{T} \right)
    \]
    Then, since our predictions are unbiased subject to the best response events of $u_{\delta'}$, Theorem \ref{thm:swapregretbound} bounds the swap regret as measured with respect to $u_{\delta'}$. For any action modification rule $\phi:\cA\to\cA$, we have:
    \[
    \sum_{t=1}^T \left( u_{\delta'} (\phi(a_t), y_t) - u_{\delta'}(a_t, y_t) \right) \leq O\left( L \sqrt{\frac{|\cA|\ln(Tdk\left(\frac{1}{\delta}\right)^{k(d+1)})}{T}} \right) \leq O\left( L \sqrt{\frac{|\cA|dk\ln(Tdk\left(\frac{1}{\delta}\right))}{T}} \right)
    \]
    The proof then follows from the fact the agent's payoff under $u$ is comparable to their payoff under $u_{\delta'}$ for any sequence of realized actions and benchmarks. We have that
    \begin{align*}
        &\max_{\phi: \cA \to \cA} \sum_{t=1}^T \left( u(\phi(a_t), y_t) - u(a_t, y_t) \right) \\
        &\leq \left( \max_{\phi: \cA \to \cA} \sum_{t=1}^T \left( u_{\delta'}(\phi(a_t), y_t) - u_{\delta'}(a_t, y_t) \right) \right) + 2\delta' \\
        &\leq O\left( L \sqrt{\frac{|\cA|dk\ln(Tdk\left(\frac{1}{\delta}\right))}{T}} \right) + 2\delta(d+1)\\
        &\leq O\left( L \sqrt{\frac{|\cA|dk\ln(Tdk)}{T}} \right)
    \end{align*}
    where in the last step we use our setting of $\delta$.
\end{proof}

\subsection{Predicting for Arbitrary Agents} \label{sec:smooth}

One might worry that it is unreasonable to expect agents to best respond using a \textit{snapped} utility function, even if it is payoff-wise close to their true utility function. Motivated by this concern, we now turn to a behavioral assumption that is  well-studied in the economics literature---namely, that agents choose their actions from a \textit{smoothed} distribution. In our main result in this section, we show how to construct a collection of events such that, using the algorithm of Theorem \ref{thm:biasbound}, downstream agents who \textit{smoothly best respond} are guaranteed diminishing swap regret at a rate of $\tilde O(T^{2/3})$.

\subsubsection{Smooth Best Response}

We first define our notion of smooth best response, which is called \textit{logistic response} in the literature. Whereas the best response function places all of the mass of an agent's response onto the utility-maximizing action, the smooth analogue assigns mass to each action in proportion to its utility. Under logistic response, an agent with action set $\cA$ and utility function $u$ plays the distribution $q(u, y) \in \Delta \cA$ where the weight placed on action $a$ is given by:
\[
q_a(u, y) = \frac{\exp(\eta \cdot u(a, y))}{\sum_{a'\in\cA} \exp(\eta \cdot u(a', y))}
\]
We should think of the parameter $\eta > 0$ as controlling the smoothness of the distribution. In particular, taking $\eta \to \infty$, logistic response coincides with strictly best responding.

As previously mentioned, smoothening the best response function harks back to an extensive line of work on quantal choice theory in the economics literature. Specifically, quantal choice models agents who best respond to utilities that are perturbed by a noise vector. The logistic response function is a popular instantiation of quantal response \citep{luce1959individual, mcfadden1976quantalchoice, mckelvey1995quantalresponse, anderson2002logit, goeree2002quantal}, where the noise vector is drawn from the Gumbel distribution which has cumulative distribution function $F(x) = \exp(-\exp(-\eta x))$; in this case, it can be shown that the probability an agent chooses an action $a$ is proportional to $\exp(\eta \cdot u(a,y))$.

Finally, note that while we prove the next result with this particular instantiation of smoothed response in mind, our result extends to any response function satisfying two conditions: first, the function gives an \textit{approximate} best response, and second, the induced weights are sufficiently \textit{Lipschitz} in utilities (the swap regret bounds will depend on the approximation factor and the Lipschitz constant). We focus on logistic response just for concreteness.

\subsubsection{Connecting Unbiased Predictions and Agent Swap Regret}

As before, to argue that agents acting according to our predictions will have diminishing swap regret according to the actual realizations, it is crucial for agents to be able to estimate their payoffs using our predictions in an unbiased way. One might hope that the technique used to derive our results in low-dimensional cases---namely, requiring our predictions to be unbiased conditional on best response events---might also apply here. However, agents now choose from a much richer action space: distributions over $\cA$. In some sense, best response events are too coarse in that they only capture the rounds on which some action was assigned the \textit{highest} probability. Instead, we would like our predictions to be unbiased conditional on the event the agent plays some distribution over $\cA$. This would be prohibitive, since there are exponentially many such distributions (up to discretization).  Fortunately, we will show that it suffices to consider only events defined by the  \textit{marginal} probabilities placed on actions---i.e. events defined by an agent placing some probability $p$ on some action $a$. This will be sufficient because of the linearity of the agents' utility functions. To construct a finite set of events, we let $S_\tau = \{0,...,\floor*{\frac{1}{\tau}}-1\}$ parameterize a bucketing of width $\tau$. Let $B^i_\tau = [i\tau, (i+1)\tau)$ for $i = 0,...,\floor*{\frac{1}{\tau}}-2$ and $B_\tau^i = [i\tau, (i+1)\tau]$ for $i = \floor*{\frac{1}{\tau}}-1$ denote the $i^{th}$ bucket. Next we define events corresponding to bucketed marginal probabilities and establish a regret bound for the general setting.  

\begin{theorem} \label{thm:quantalresponse}
    Fix a transcript $\pi_T$. Let $\cU^k$ be the set of utility functions over an action set of size at most $k$ and prediction space $\cC$ that are linear in its second argument. Let $\cU_{\delta'}^k$ be the $\delta'$-cover of $\cU^k$ given by Lemma \ref{lem:cover}, where $\delta' = \delta(d+1)$ for some parameter $\delta > 0$. Consider an agent with action set $\cA$ and utility function $u \in \cU^k$ who, at every round $t \in [T]$, plays the distribution $q(u, \hat{y}_t) \in \Delta \cA$ given by the logistic response function specified by some parameter $\eta > 0$. Let $q_{a}(u, \hat{y}_t)$ be the weight assigned to action $a$. Let $E_{u, a, i}(y) = \1[q_{a}(u, y) \in B^i_\tau]$. Define $\cE$ to be the collection of events $\cE = \{ E_{u_{\delta'}, a, i} : u_{\delta'} \in \cU_{\delta'}^k, a \in [k], i \in S_\tau \}$. Then, using the algorithm given by \citet{noarov2023highdimensional} parameterized by event set $\cE$ to make predictions, any such agent is guaranteed swap regret bounded by:
    \[
    \max_{\phi: \cA \to \cA} Reg(u, \phi) \leq O\left( \frac{|\cA|L \sqrt{Ldk \ln(TLdk)}}{T^{1/3}} \right)
    \] 
    for $\eta = O(\sqrt{T}\ln k), \delta = O\left(\frac{\ln(1/k\sqrt{T})}{d\sqrt{T}}\right)$, and $\tau = O\left(\frac{1}{kLT^{1/3}}\right)$.
\end{theorem}
\begin{proof}
    Recall that the guarantees of conditional bias given by \citet{noarov2023highdimensional} depend on the number of events. By Lemma \ref{lem:cover}, we have that $|\cE| = k|S_\tau||\cU_{\delta'}^k| \leq k \floor*{\frac{1}{\tau}} \left(\frac{1}{\delta}\right)^{k(d+1)}$. Plugging this into Theorem \ref{thm:biasbound}, we can conclude that for sufficiently small $\eps$, the conditional bias on $\cE$ is at most 
    \begin{align*}
        \alpha(n_T(E_{u, a, i})) &\leq O\left( \frac{\ln(Tdk\floor*{\frac{1}{\tau}} \left(\frac{1}{\delta}\right)^{k(d+1)}) + \sqrt{\ln(Tdk\floor*{\frac{1}{\tau}} \left(\frac{1}{\delta}\right)^{k(d+1)}) \cdot n_T(E_{u, a, i})}}{T} \right) \\
        &\leq O\left( \frac{dk\ln(Tdk\floor*{\frac{1}{\tau}} \left(\frac{1}{\delta}\right)) + \sqrt{dk\ln(Tdk\floor*{\frac{1}{\tau}} \left(\frac{1}{\delta}\right)) \cdot n_T(E_{u, a, i})}}{T} \right)
    \end{align*} 
    
    The remainder of the proof is structured as follows. Recall that an agent's utility function $u$ can be closely approximated by a nearby utility function $u_{\delta'}$ in the $\delta'$-cover $\cU_{\delta'}^k$. We first measure the swap regret of the agent \textit{had they behaved and received payoff according to $u_{\delta'}$}. Here, we rely on the fact that the predictions are unbiased conditional on (bucketed) marginal probabilities assigned to each action. Thus, the agent's estimate of their utility whenever they assign a certain probability to an action is (on average) correct, and so the agent's calculation of their expected utility over the distribution they play---and any other distribution---is correct. Therefore, since the agent plays an approximate best response, the distribution they play yields no swap regret up to the bias of our predictions \textit{and} the gap introduced by the approximation.

    Of course, the agent actually behaves and receives payoff according to $u$ rather than $u_{\delta'}$. Here, we use the fact that the logistic response function is Lipschitz in utilities---in particular, since $u$ and $u_{\delta'}$ give similar utilities, $u$ and $u_{\delta'}$ induce similar distributions over actions. Thus, we can establish that the expected utility the agent obtains by responding according to $u$ cannot be too much smaller than the expected utility the agent could have obtained by responding according to $u_{\delta'}$. Similarly, the expected utility obtained by any benchmark class of actions under $u$ cannot be too much larger than the expected utility obtained by the same benchmark class under $u_{\delta'}$. This then bounds the swap regret of the agent.
    
    Thus, to prove Theorem \ref{thm:quantalresponse}, we first introduce the following intermediate lemma that bounds the swap regret as measured with respect to $u_{\delta'}$. 

    \begin{lemma}\label{lem:quantalresponse}
        Fix a utility function $u_{\delta'} \in \cU_{\delta'}^k$ and consider an agent who, at round $t$, plays the distribution $q(u_{\delta'}, \hat{y}_t)$ given by the logistic response function. Then, for any action modification rule $\phi: \cA \to \cA$, we have:
        \[
        \frac{1}{T} \sum_{t=1}^T \E_{\hat{y}_t \sim p_t} \E_{a \sim q(u_{\delta'}, \hat{y}_t)}[u_{\delta'}(\phi(a), y_t) - u_{\delta'}(a, y_t)] \leq \frac{\ln|\cA|+1}{\eta} + 2|\cA|L\floor*{\frac{1}{\tau}} \alpha\left(T/\floor*{\frac{1}{\tau}}\right) + 2|\cA|L\tau
        \]
    \end{lemma}
    The proof of Lemma \ref{lem:quantalresponse} will make use of two lemmas. The first lemma will help us establish that in expectation over the rounds in which the agent plays action $a$, our conditional bias guarantees are (approximately) preserved.
    
    \begin{lemma}\label{lem:weightedbias}
        Fix any $a \in \cA$. If $\pi_T$ results in bias at most $\alpha(n_T(E_{u_{\delta'}, a, i}))$ conditional on the event $E_{u_{\delta'}, a, i}$ for all $i \in S_\tau$, then we have:
        \[
        \left\| \frac{1}{T} \sum_{t=1}^T \sum_{i\in S_\tau} \E_{\hat{y_t}\sim p_t}[E_{u_{\delta'}, a, i}(\hat{y_t}) q_a(u_{\delta'}, \hat{y}_t) (\hat{y_t} - y_t)] \right\|_\infty \leq \floor*{\frac{1}{\tau}} \alpha\left(T/\floor*{\frac{1}{\tau}}\right) + \tau
        \]
    \end{lemma}
    \begin{proof}
        Observe that
        \begin{align*}
            & \left\| \frac{1}{T} \sum_{t=1}^T \sum_{i\in S_\tau} \E_{\hat{y_t}\sim p_t}[E_{u_{\delta'}, a, i}(\hat{y_t}) q_a(u_{\delta'}, \hat{y}_t) (\hat{y_t} - y_t)] \right\|_\infty \\
            &= \Bigg\| \frac{1}{T} \sum_{t=1}^T \sum_{i\in S_\tau} \E_{\hat{y_t}\sim p_t}[\1[(\hat{y_t} - y_t) \geq 0] \1[q_a(u_{\delta'}, \hat{y}_t) \in B_\tau^i] q_a(u_{\delta'}, \hat{y}_t) (\hat{y_t} - y_t) \\
            & \hspace{7em} + \1[(\hat{y_t} - y_t) < 0] \1[q_a(u_{\delta'}, \hat{y}_t) \in B_\tau^i] q_a(u_{\delta'}, \hat{y}_t) (\hat{y_t} - y_t)] \Bigg\|_\infty \\
            & \leq \Bigg\| \frac{1}{T} \sum_{t=1}^T \sum_{i\in S_\tau} \E_{\hat{y_t}\sim p_t}[\1[(\hat{y_t} - y_t) \geq 0] \1[q_a(u_{\delta'}, \hat{y}_t) \in B_\tau^i] \cdot ((i+1)\tau) (\hat{y_t} - y_t) \\
            & \hspace{7em} + \1[(\hat{y_t} - y_t) < 0] \1[q_a(u_{\delta'}, \hat{y}_t) \in B_\tau^i] \cdot (i\tau) (\hat{y_t} - y_t)] \Bigg\|_\infty \\
            &\leq \left\| \frac{1}{T} \sum_{t=1}^T \sum_{i\in S_\tau} \E_{\hat{y_t}\sim p_t}[\1[q_a(u_{\delta'}, \hat{y}_t) \in B_\tau^i]\cdot (i\tau) (\hat{y_t} - y_t) + \1[q_a(u_{\delta'}, \hat{y}_t) \in B_\tau^i] \cdot \tau] \right\|_\infty \\
            &= \left\| \frac{1}{T} \sum_{i\in S_\tau} i\tau  \sum_{t=1}^T \E_{\hat{y_t}\sim p_t}\left[ \1[q_a(u_{\delta'}, \hat{y}_t) \in B_\tau^i] (\hat{y_t} - y_t)\right] \right\|_\infty + \frac{\tau}{T} \sum_{t=1}^T \E_{\hat{y_t}\sim p_t}\left[\sum_{i\in S_\tau} \1[q_a(u_{\delta'}, \hat{y}_t) \in B_\tau^i] \right]  \\
            &\leq \sum_{i\in S_\tau} i\tau \cdot \frac{1}{T} \left\| \sum_{t=1}^T \E_{\hat{y_t}\sim p_t}\left[ \1[q_a(u_{\delta'}, \hat{y}_t) \in B_\tau^i] (\hat{y_t} - y_t)\right] \right\|_\infty + \frac{\tau}{T} \sum_{t=1}^T \E_{\hat{y_t}\sim p_t}\left[\sum_{i\in S_\tau} \1[q_a(u_{\delta'}, \hat{y}_t) \in B_\tau^i] \right]  \\
            &\leq \sum_{i\in S_\tau}\alpha(n_T(E_{u_{\delta'},a,i})) + \tau
        \end{align*}
        The last inequality follows from our assumption of $\alpha$-unbiasedness, the fact that $i\tau \leq 1$ for all $i$, and the fact that for any utility function $u_{\delta'}$ and action $a$, the events $\{E_{u_{\delta'},a,i}\}_{i\in S_\tau}$ are disjoint---the indicator $\1[q_a(u_{\delta'}, \hat{y}_t) \in B_\tau^i] = 1$ for exactly one $i \in S_\tau$, i.e. $\sum_{i\in S_\tau} E_{u_{\delta'},a,i}(\hat{y}_t) = 1$. Thus, it also follows that $\sum_{i\in S_\tau} n_T(E_{u_{\delta'},a,i}) = T$, and so by concavity of $\alpha$, we have that $\sum_{i\in S_\tau}\alpha(n_T(E_{u_{\delta'},a,i})) \leq |S_\tau| \alpha((\sum_{i\in S_\tau}n_T(E_{u_{\delta'},a,i}))/|S_\tau|) = |S_\tau| \alpha(T/|S_\tau|)$, which completes the proof.
    \end{proof}

    The next lemma establishes that the utility gained by playing the logistic response distribution cannot be too much smaller than the utility gained by playing the best response action (note that this is a general statement that applies to any utility function). The difference between the two quantities will be controlled by the parameter $\eta$. 

    \begin{lemma}\label{lem:quantalresponseutility}
        Fix any utility function $u$ and any $y \in \cC$. Let $a^* = BR(u, y)$. We have that:
        \[
        \E_{a\sim q(u, y)}[u(a, y)] \geq u(a^*, y) - \frac{\ln|\cA| + 1}{\eta}
        \]
    \end{lemma}
    \begin{proof}
        First, we see that for any constant $x$,
        \begin{align*}
            \Pr[u(a,y) \leq x] &\leq \frac{\Pr[u(a,y) \leq x]}{\Pr[u(a,y) = u(a^*,y)]} \leq \frac{|\cA|\exp(\eta x)}{\exp(\eta u(a^*, y)} = |\cA| \exp(\eta(x - u(a^*,y)))
        \end{align*}
        Taking $x = u(a^*,y) - \frac{1}{\eta}(\ln|\cA| + c)$, we have that
        \begin{align*}
            \Pr[u(a,y) \leq u(a^*,y) - \frac{1}{\eta}(\ln|\cA| + c)] \leq |\cA| e^{-\ln|\cA|-c} = e^{-c}
        \end{align*}
        or equivalently, $\Pr[\eta(u(a^*,y) - u(a,y)) \geq \ln|\cA|+c] \leq e^{-c}$. Let $X = \eta(u(a^*,y) - u(a,y))$. $X$ is a non-negative random variable, so we have
        \begin{align*}
            \E[X] = \int_0^\infty \Pr[X \geq x] dx = \int_{c = -\ln|\cA|}^\infty \Pr[X \geq \ln|\cA| + c] dc \leq \int_{c = -\ln|\cA|}^0 1 dc + \int_0^\infty e^{-c} dc = \ln|\cA| + 1
        \end{align*}
        where the inequality uses the fact that  $\Pr[X \geq \ln|\cA|+c] \leq 1$ for $c \leq 0$. Hence it follows that $u(a^*,y) - \E[u(a,y)] \leq \frac{\ln|\cA|+1}{\eta}$, which proves the lemma.
    \end{proof}

    We are now ready to prove Lemma \ref{lem:quantalresponse}. 

    \begin{proofof}{Lemma \ref{lem:quantalresponse}}
        Fix any action modification rule $\phi: \cA \to \cA$. Using linearity of $u_{\delta'}$ in the second argument, we derive:
        \begin{align*}
            & \frac{1}{T} \sum_{t=1}^T \E_{\hat{y}_t \sim p_t} \E_{a \sim q(u_{\delta'}, \hat{y}_t)}[u_{\delta'}(\phi(a), y_t) - u_{\delta'}(a, y_t)] \\
            &= \sum_{a\in\cA} \frac{1}{T} \sum_{t=1}^T \sum_{i \in S_\tau} \E_{\hat{y}_t \sim p_t} \left[ E_{u_{\delta'},a,i}(\hat{y}_t) q_a(u_{\delta'}, \hat{y}_t) (u_{\delta'}(\phi(a), y_t) - u_{\delta'}(a, y_t)) \right] \\
            &= \sum_{a\in\cA} \Bigg( u_{\delta'}(\phi(a), \frac{1}{T} \sum_{t=1}^T \sum_{i \in S_\tau} \E_{\hat{y}_t \sim p_t}[E_{u_{\delta'},a,i}(\hat{y}_t) q_a(u_{\delta'}, \hat{y}_t)]y_t) \\
            & \hspace{2.5em} - u_{\delta'}(a, \frac{1}{T} \sum_{t=1}^T \sum_{i \in S_\tau} \E_{\hat{y}_t \sim p_t}[E_{u_{\delta'},a,i}(\hat{y}_t) q_a(u_{\delta'}, \hat{y}_t)]y_t) \Bigg) \\
            &\leq \Bigg( \sum_{a\in\cA} \Bigg( u_{\delta'}(\phi(a), \frac{1}{T} \sum_{t=1}^T \sum_{i \in S_\tau} \E_{\hat{y}_t \sim p_t}[E_{u_{\delta'},a,i}(\hat{y}_t) q_a(u_{\delta'}, \hat{y}_t) \hat{y}_t]) \\
            & \hspace{3em} - u_{\delta'}(a, \frac{1}{T} \sum_{t=1}^T \sum_{i \in S_\tau} \E_{\hat{y}_t \sim p_t}[E_{u_{\delta'},a,i}(\hat{y}_t) q_a(u_{\delta'}, \hat{y}_t) \hat{y}_t])\Bigg) \Bigg) + 2|\cA|L\floor*{\frac{1}{\tau}} \alpha\left(T/\floor*{\frac{1}{\tau}}\right) + 2|\cA|L\tau \\
            &= \left( \sum_{a\in\cA} \frac{1}{T} \sum_{t=1}^T \sum_{i \in S_\tau} \E_{\hat{y}_t \sim p_t} \left[ E_{u_{\delta'},a,i}(\hat{y}_t) q_a(u_{\delta'}, \hat{y}_t) (u_{\delta'}(\phi(a), \hat{y}_t) - u_{\delta'}(a, \hat{y}_t)) \right] \right) \\
            & \hspace{1em} + 2|\cA|L\floor*{\frac{1}{\tau}} \alpha\left(T/\floor*{\frac{1}{\tau}}\right) + 2|\cA|L\tau \\
            &= \left( \frac{1}{T} \sum_{t=1}^T \E_{\hat{y}_t \sim p_t} \E_{a \sim q(u_{\delta'}, \hat{y}_t)}\left[ u_{\delta'}(\phi(a), \hat{y}_t) - u_{\delta'}(a, \hat{y}_t) \right] \right) + 2|\cA|L\floor*{\frac{1}{\tau}} \alpha\left(T/\floor*{\frac{1}{\tau}}\right) + 2|\cA|L\tau \\
            \intertext{which, by the fact that the benchmark action $\phi(a)$ cannot obtain higher utility than the best response at round $t$, letting $a^*_t = BR(u_{\delta'}, \hat{y}_t)$, this quantity is at most}
            &\leq \left( \frac{1}{T} \sum_{t=1}^T \E_{\hat{y}_t \sim p_t}[u_{\delta'}(a^*_t, \hat{y}_t) - \E_{a \sim q(u_{\delta'}, \hat{y}_t)}\left[ u_{\delta'}(a, \hat{y}_t) \right] ] \right) + 2|\cA|L\floor*{\frac{1}{\tau}} \alpha\left(T/\floor*{\frac{1}{\tau}}\right) + 2|\cA|L\tau \\
            &\leq \left( \frac{1}{T} \sum_{t=1}^T \frac{\ln|\cA|+1}{\eta} \right) + 2|\cA|L\floor*{\frac{1}{\tau}} \alpha\left(T/\floor*{\frac{1}{\tau}}\right) + 2|\cA|L\tau \\
            &= \frac{\ln|\cA|+1}{\eta} + 2|\cA|L\floor*{\frac{1}{\tau}} \alpha\left(T/\floor*{\frac{1}{\tau}}\right) + 2|\cA|L\tau
        \end{align*}
        Here we apply Lemma \ref{lem:weightedbias} and $L$-Lipschitzness of $u_{\delta'}$ to get the first inequality: for any action $a$,
        \begin{align*}
            &\left|u_{\delta'}(a, \frac{1}{T} \sum_{t=1}^T \sum_{i \in S_\tau} \E_{\hat{y}_t \sim p_t}[E_{u_{\delta'},a,i}(\hat{y}_t) q_a(u_{\delta'}, \hat{y}_t) \hat{y}_t]) - u_{\delta'}(a, \frac{1}{T} \sum_{t=1}^T \sum_{i \in S_\tau} \E_{\hat{y}_t \sim p_t}[E_{u_{\delta'},a,i}(\hat{y}_t) q_a(u_{\delta'}, \hat{y}_t)]y_t)\right| \\
            &\leq L \left\| \frac{1}{T} \sum_{t=1}^T \sum_{i \in S_\tau} \E_{\hat{y}_t \sim p_t}[E_{u_{\delta'},a,i}(\hat{y}_t) q_a(u_{\delta'}, \hat{y}_t) (\hat{y}_t - y_t)] \right\|_\infty \\
            &\leq L \floor*{\frac{1}{\tau}} \alpha\left(T/\floor*{\frac{1}{\tau}}\right) + L\tau
        \end{align*}
        Finally, we apply Lemma \ref{lem:quantalresponseutility} to get the last inequality. This completes the proof.
    \end{proofof}

    Before proving Theorem \ref{thm:quantalresponse}, we introduce one more lemma which states that nearby utility functions induce nearby distributions over actions---in particular, perturbing the payoff of any action by a little bit does not change its assigned weight by too much.

        \begin{lemma}\label{lem:quantalresponsesmooth}
            Fix any $a\in\cA, y\in\cC$. For any utility functions $u_1, u_2$, if $|u_1(a,y) - u_2(a,y)| \leq \delta$, then $|q_a(u_1, y) - q_a(u_2, y)| \leq e^{2\eta\delta} - 1$. 
        \end{lemma}
        \begin{proof}
            W.l.o.g. assume that $u_1(a,y) \geq u_2(a,y)$.  Then we have that
            \begin{align*}
                \frac{q_a(u_1, y)}{q_a(u_2, y)} &= \frac{\left(\frac{\exp(\eta u_1(a,y))}{\sum_{a'\in\cA} \exp(\eta u_1(a',y))}\right)}{\left(\frac{\exp(\eta u_2(a,y))}{\sum_{a'\in\cA} \exp(\eta u_2(a',y))}\right)} \\
                &= \exp(\eta (u_1(a,y) - u_2(a,y))) \cdot \left(\frac{\sum_{a'\in\cA} \exp(\eta u_2(a',y))}{\sum_{a'\in\cA} \exp(\eta u_1(a',y))}\right) \\
                &\leq \exp(\eta\delta) \cdot \left(\frac{\sum_{a'\in\cA} \exp(\eta (u_1(a',y)+\delta)}{\sum_{a'\in\cA} \exp(\eta u_1(a',y))}\right) \\
                &= \exp(\eta\delta) \cdot \exp(\eta\delta) \\
                &= \exp(2\eta\delta)
            \end{align*}
            Rearranging the expression, we get $q_a(u_1, y) \leq \exp(2\eta\delta) q_a(u_2, y)$ and hence, since $q_a(u_2, y) \leq 1$, we have that
            \[
            q_a(u_1, y) - q_a(u_2, y) \leq \exp(2\eta\delta) q_a(u_2, y) - q_a(u_2, y) \leq \exp(2\eta\delta) - 1
            \]
            which proves the lemma. 
        \end{proof}

    We now proceed with the proof of Theorem \ref{thm:quantalresponse}. Recall that by definition of a $\delta'$-net, we can find $u_{\delta'} \in \cU_{\delta'}^k$ such that $|u(a,y) - u_{\delta'}(a,y)| \leq \delta'$ for all $a\in\cA, y\in\cC$. Hence, together with Lemma \ref{lem:quantalresponsesmooth}, we can relate the cumulative utility of the agent to their cumulative utility if they instead had utility function $u_{\delta'}$---for both the realized and benchmark sequence of actions. An application of Lemma \ref{lem:quantalresponse} will then bound the agent's swap regret. In particular, for any $\phi: \cA\to\cA$, we have that 
        \begin{align*}
            Reg(u, \phi) &= \frac{1}{T} \sum_{t=1}^T \E_{\hat{y}_t \sim p_t} \E_{a \sim q(u, \hat{y}_t)}[u(\phi(a), y_t) - u(a, y_t)] \\
            &\leq \left( \frac{1}{T} \sum_{t=1}^T \E_{\hat{y}_t \sim p_t} \E_{a \sim q(u, \hat{y}_t)}[u_{\delta'}(\phi(a), y_t) - u_{\delta'}(a, y_t)] \right) + 2\delta' \\
            &= \left( \frac{1}{T} \sum_{t=1}^T \sum_{a\in\cA} \E_{\hat{y}_t \sim p_t} [q_a(u, \hat{y}_t) (u_{\delta'}(\phi(a), y_t) - u_{\delta'}(a, y_t))] \right) + 2\delta' \\
            &\leq \left( \frac{1}{T} \sum_{t=1}^T \sum_{a\in\cA} \E_{\hat{y}_t \sim p_t} [q_a(u_{\delta'}, \hat{y}_t) (u_{\delta'}(\phi(a), y_t) - u_{\delta'}(a, y_t))] \right) + |\cA|(e^{2\eta\delta'} - 1) + 2\delta' \\
            &= \left( \frac{1}{T} \sum_{t=1}^T \E_{\hat{y}_t \sim p_t} \E_{a \sim q(u_{\delta'}, \hat{y}_t)}[u_{\delta'}(\phi(a), y_t) - u_{\delta'}(a, y_t)] \right) + |\cA|(e^{2\eta\delta'} - 1) + 2\delta' \\
            &\leq \frac{\ln|\cA|+1}{\eta} + 2|\cA|L\floor*{\frac{1}{\tau}} \alpha\left(T/\floor*{\frac{1}{\tau}}\right) + 2|\cA|L\tau + |\cA|(e^{2\eta\delta'} - 1) + 2\delta' \\
            &\leq \frac{\ln|\cA|+1}{\eta} + O\left( |\cA|L \sqrt{\frac{\floor*{\frac{1}{\tau}} dk \ln(Tdk\floor*{\frac{1}{\tau}} \left(\frac{1}{\delta}\right))}{T}} \right) + 2|\cA|L\tau + |\cA|(e^{2\eta\delta'} - 1) + 2\delta'
        \end{align*}
        where the first inequality follows from our choice of $u_{\delta'}$, the second inequality follows from Lemma \ref{lem:quantalresponsesmooth}, and the last inequality follows from Lemma \ref{lem:quantalresponse}. In the last step, we substitute in the conditional bias bound established previously. For $\eta = (\ln k + 1)\sqrt{T}, \delta = \frac{\ln((1/k\sqrt{T})+1)}{(d+1)\sqrt{T}}, \tau = \frac{1}{kLT^{1/3}}$, we get
        \[
        Reg(u, \phi) \leq O\left( \frac{|\cA|L \sqrt{Ldk \ln(TLdk)}}{T^{1/3}} \right)
        \]

\end{proof}

\section{Discussion and Conclusion}
We have shown that it is possible to make predictions that guarantee diminishing swap regret to \emph{all} downstream agents simultaneously, at rates that are substantially faster than are possible by making ($\ell_1$)-calibrated predictions. In particular, for best responding agents, we have shown how to do this at the optimal $O(\sqrt{T})$ rate for $d = 1$ --- a rate that is impossible to obtain via calibration. For agents that are assumed to best respond according to a discretized version of their own utility function, we've shown how to obtain these optimal $O(\sqrt{T})$ rates for \emph{every} dimension $d$, and for agents that \emph{smoothly} best respond according to their own utility function, we've shown how to guarantee all downstream agents swap regret at the dimension-independent rate $O(T^{2/3})$ in $d$ dimensions. Thus we have established---perhaps surprisingly---that a substantially weaker condition than calibration suffices, even when the goal is to guarantee downstream swap regret simultaneously for all agents. Prior work was able to obtain similar guarantees only by either relaxing from swap to external regret \citep{kleinberg2023ucalibration} or by restricting to a fixed collection of downstream agents, rather than offering guarantees simultaneously to all downstream agents \citep{noarov2023highdimensional}. It is natural to want to use techniques emerging from this literature to coordinate players in a repeated game to play correlated equilibria. In such settings, the state being predicted at each round is a function of the actions chosen by all of the players in the game at that round. The results from this literature do not directly apply to this case, because when the state is a function of player actions, which are themselves a function of the predictions of the forecaster, the state depends on the realized randomness of the forecaster, which is not allowed in the adversarial model in which the theorems are proven. A simple fix is to find a single distribution on forecasts for all downstream agents, but take independent samples from this distribution to transmit to each agent in the game. This works well with our approach, as the computationally difficult step is computing the sampling distribution, which only needs to be done once. 

Unlike the bounds of \cite{kleinberg2023ucalibration}, our bounds have a dependence on the number of actions available to the downstream agent. However, in the special case of $d = 1$, \cite{hu2024predict} have shown how to remove this dependence. This suggests a natural open problem: For $d > 1$, is a dependence on the number of actions necessary in order to obtain dimension-independent swap regret bounds, or can it be removed? 

Although our approach is computationally efficient for $d=1$, like algorithms promising calibration, the computational complexity of our approach scales badly with $d$. We leave as the main open question from our work:
\begin{quote}
    Is there an algorithm that can make $d$ dimensional predictions that guarantee all downstream agents swap regret diminishing at a rate of $\tilde O(T^{O(1)})$ with per-round running time scaling polynomially with $d$?
\end{quote}

\subsection*{Acknowledgements}
We thank Georgy Noarov for enlightening conversations and for pointing us to \citet{ivic1994digital}. This work was supported in part by the Simons Collaboration on the Theory of Algorithmic Fairness, NSF grants FAI-2147212 and CCF-2217062, an AWS AI Gift for Research on Trustworthy AI, and the Hans Sigrist Prize.

\bibliographystyle{ACM-Reference-Format}
\bibliography{main}


\begin{thebibliography}{26}


\ifx \showCODEN    \undefined \def \showCODEN     #1{\unskip}     \fi
\ifx \showDOI      \undefined \def \showDOI       #1{#1}\fi
\ifx \showISBNx    \undefined \def \showISBNx     #1{\unskip}     \fi
\ifx \showISBNxiii \undefined \def \showISBNxiii  #1{\unskip}     \fi
\ifx \showISSN     \undefined \def \showISSN      #1{\unskip}     \fi
\ifx \showLCCN     \undefined \def \showLCCN      #1{\unskip}     \fi
\ifx \shownote     \undefined \def \shownote      #1{#1}          \fi
\ifx \showarticletitle \undefined \def \showarticletitle #1{#1}   \fi
\ifx \showURL      \undefined \def \showURL       {\relax}        \fi
\providecommand\bibfield[2]{#2}
\providecommand\bibinfo[2]{#2}
\providecommand\natexlab[1]{#1}
\providecommand\showeprint[2][]{arXiv:#2}

\bibitem[Anderson et~al\mbox{.}(2002)]%
        {anderson2002logit}
\bibfield{author}{\bibinfo{person}{Simon~P. Anderson}, \bibinfo{person}{Jacob~K. Goeree}, {and} \bibinfo{person}{Charles~A. Holt}.} \bibinfo{year}{2002}\natexlab{}.
\newblock \showarticletitle{The Logit Equilibrium: A Perspective on Intuitive Behavioral Anomalies}.
\newblock \bibinfo{journal}{\emph{Southern Economic Journal}} \bibinfo{volume}{69}, \bibinfo{number}{1} (\bibinfo{year}{2002}), \bibinfo{pages}{21--47}.
\newblock
\showISSN{00384038}
\urldef\tempurl%
\url{http://www.jstor.org/stable/1061555}
\showURL{%
\tempurl}


\bibitem[Blum and Mansour(2007)]%
        {blum2007internal}
\bibfield{author}{\bibinfo{person}{Avrim Blum} {and} \bibinfo{person}{Yishay Mansour}.} \bibinfo{year}{2007}\natexlab{}.
\newblock \showarticletitle{From External to Internal Regret}.
\newblock \bibinfo{journal}{\emph{Journal of Machine Learning Research}} \bibinfo{volume}{8}, \bibinfo{number}{47} (\bibinfo{year}{2007}), \bibinfo{pages}{1307--1324}.
\newblock
\urldef\tempurl%
\url{http://jmlr.org/papers/v8/blum07a.html}
\showURL{%
\tempurl}


\bibitem[Dawid(1982)]%
        {dawid1982well}
\bibfield{author}{\bibinfo{person}{A.~P. Dawid}.} \bibinfo{year}{1982}\natexlab{}.
\newblock \showarticletitle{The Well-Calibrated Bayesian}.
\newblock \bibinfo{journal}{\emph{J. Amer. Statist. Assoc.}} \bibinfo{volume}{77}, \bibinfo{number}{379} (\bibinfo{year}{1982}), \bibinfo{pages}{605--610}.
\newblock
\urldef\tempurl%
\url{https://doi.org/10.1080/01621459.1982.10477856}
\showDOI{\tempurl}


\bibitem[Foster and Hart(2018)]%
        {foster2018smooth}
\bibfield{author}{\bibinfo{person}{Dean~P. Foster} {and} \bibinfo{person}{Sergiu Hart}.} \bibinfo{year}{2018}\natexlab{}.
\newblock \showarticletitle{Smooth calibration, leaky forecasts, finite recall, and Nash dynamics}.
\newblock \bibinfo{journal}{\emph{Games and Economic Behavior}}  \bibinfo{volume}{109} (\bibinfo{year}{2018}), \bibinfo{pages}{271--293}.
\newblock
\showISSN{0899-8256}
\urldef\tempurl%
\url{https://doi.org/10.1016/j.geb.2017.12.022}
\showDOI{\tempurl}


\bibitem[Foster and Vohra(1997)]%
        {foster1997calibrated}
\bibfield{author}{\bibinfo{person}{Dean~P. Foster} {and} \bibinfo{person}{Rakesh~V. Vohra}.} \bibinfo{year}{1997}\natexlab{}.
\newblock \showarticletitle{Calibrated Learning and Correlated Equilibrium}.
\newblock \bibinfo{journal}{\emph{Games and Economic Behavior}} \bibinfo{volume}{21}, \bibinfo{number}{1} (\bibinfo{year}{1997}), \bibinfo{pages}{40--55}.
\newblock
\showISSN{0899-8256}
\urldef\tempurl%
\url{https://doi.org/10.1006/game.1997.0595}
\showDOI{\tempurl}


\bibitem[Foster and Vohra(1998)]%
        {foster1998asymptotic}
\bibfield{author}{\bibinfo{person}{Dean~P. Foster} {and} \bibinfo{person}{Rakesh~V. Vohra}.} \bibinfo{year}{1998}\natexlab{}.
\newblock \showarticletitle{Asymptotic Calibration}.
\newblock \bibinfo{journal}{\emph{Biometrika}} \bibinfo{volume}{85}, \bibinfo{number}{2} (\bibinfo{year}{1998}), \bibinfo{pages}{379--390}.
\newblock
\showISSN{00063444}
\urldef\tempurl%
\url{http://www.jstor.org/stable/2337364}
\showURL{%
\tempurl}


\bibitem[Garg et~al\mbox{.}(2024)]%
        {garg2024oracle}
\bibfield{author}{\bibinfo{person}{Sumegha Garg}, \bibinfo{person}{Christopher Jung}, \bibinfo{person}{Omer Reingold}, {and} \bibinfo{person}{Aaron Roth}.} \bibinfo{year}{2024}\natexlab{}.
\newblock \bibinfo{booktitle}{\emph{Oracle Efficient Online Multicalibration and Omniprediction}}.
\newblock \bibinfo{pages}{2725--2792}.
\newblock
\urldef\tempurl%
\url{https://doi.org/10.1137/1.9781611977912.98}
\showDOI{\tempurl}
\showeprint{https://epubs.siam.org/doi/pdf/10.1137/1.9781611977912.98}


\bibitem[Globus-Harris et~al\mbox{.}(2023)]%
        {globus2023multicalibrated}
\bibfield{author}{\bibinfo{person}{Ira Globus-Harris}, \bibinfo{person}{Varun Gupta}, \bibinfo{person}{Christopher Jung}, \bibinfo{person}{Michael Kearns}, \bibinfo{person}{Jamie Morgenstern}, {and} \bibinfo{person}{Aaron Roth}.} \bibinfo{year}{2023}\natexlab{}.
\newblock \showarticletitle{Multicalibrated regression for downstream fairness}. In \bibinfo{booktitle}{\emph{Proceedings of the 2023 AAAI/ACM Conference on AI, Ethics, and Society}}. \bibinfo{pages}{259--286}.
\newblock


\bibitem[Goeree et~al\mbox{.}(2002)]%
        {goeree2002quantal}
\bibfield{author}{\bibinfo{person}{Jacob~K. Goeree}, \bibinfo{person}{Charles~A. Holt}, {and} \bibinfo{person}{Thomas~R. Palfrey}.} \bibinfo{year}{2002}\natexlab{}.
\newblock \showarticletitle{Quantal Response Equilibrium and Overbidding in Private-Value Auctions}.
\newblock \bibinfo{journal}{\emph{Journal of Economic Theory}} \bibinfo{volume}{104}, \bibinfo{number}{1} (\bibinfo{year}{2002}), \bibinfo{pages}{247--272}.
\newblock
\showISSN{0022-0531}
\urldef\tempurl%
\url{https://doi.org/10.1006/jeth.2001.2914}
\showDOI{\tempurl}


\bibitem[Gopalan et~al\mbox{.}(2023)]%
        {gopalan2023loss}
\bibfield{author}{\bibinfo{person}{Parikshit Gopalan}, \bibinfo{person}{Lunjia Hu}, \bibinfo{person}{Michael~P Kim}, \bibinfo{person}{Omer Reingold}, {and} \bibinfo{person}{Udi Wieder}.} \bibinfo{year}{2023}\natexlab{}.
\newblock \showarticletitle{Loss Minimization Through the Lens Of Outcome Indistinguishability}. In \bibinfo{booktitle}{\emph{14th Innovations in Theoretical Computer Science Conference (ITCS 2023)}}. Schloss Dagstuhl-Leibniz-Zentrum f{\"u}r Informatik.
\newblock


\bibitem[Gopalan et~al\mbox{.}(2022)]%
        {gopalan2022omnipredictors}
\bibfield{author}{\bibinfo{person}{Parikshit Gopalan}, \bibinfo{person}{Adam~Tauman Kalai}, \bibinfo{person}{Omer Reingold}, \bibinfo{person}{Vatsal Sharan}, {and} \bibinfo{person}{Udi Wieder}.} \bibinfo{year}{2022}\natexlab{}.
\newblock \showarticletitle{Omnipredictors}. In \bibinfo{booktitle}{\emph{13th Innovations in Theoretical Computer Science Conference (ITCS 2022)}}. Schloss Dagstuhl-Leibniz-Zentrum f{\"u}r Informatik.
\newblock


\bibitem[Haghtalab et~al\mbox{.}(2023)]%
        {haghtalab2023calibrated}
\bibfield{author}{\bibinfo{person}{Nika Haghtalab}, \bibinfo{person}{Chara Podimata}, {and} \bibinfo{person}{Kunhe Yang}.} \bibinfo{year}{2023}\natexlab{}.
\newblock \showarticletitle{Calibrated Stackelberg Games: Learning Optimal Commitments Against Calibrated Agents}. In \bibinfo{booktitle}{\emph{Thirty-seventh Conference on Neural Information Processing Systems}}.
\newblock
\urldef\tempurl%
\url{https://openreview.net/forum?id=fHsBNNDroC}
\showURL{%
\tempurl}


\bibitem[Hu et~al\mbox{.}(2023)]%
        {hu2023omnipredictors}
\bibfield{author}{\bibinfo{person}{Lunjia Hu}, \bibinfo{person}{Inbal Rachel~Livni Navon}, \bibinfo{person}{Omer Reingold}, {and} \bibinfo{person}{Chutong Yang}.} \bibinfo{year}{2023}\natexlab{}.
\newblock \showarticletitle{Omnipredictors for constrained optimization}. In \bibinfo{booktitle}{\emph{International Conference on Machine Learning}}. PMLR, \bibinfo{pages}{13497--13527}.
\newblock


\bibitem[Hu and Wu(2024)]%
        {hu2024predict}
\bibfield{author}{\bibinfo{person}{Lunjia Hu} {and} \bibinfo{person}{Yifan Wu}.} \bibinfo{year}{2024}\natexlab{}.
\newblock \bibinfo{title}{Predict to Minimize Swap Regret for All Payoff-Bounded Tasks}.
\newblock
\newblock
\showeprint[arxiv]{2404.13503}~[cs.LG]


\bibitem[Ivic et~al\mbox{.}(1994)]%
        {ivic1994digital}
\bibfield{author}{\bibinfo{person}{A. Ivic}, \bibinfo{person}{J. Koplowitz}, {and} \bibinfo{person}{J. Zunic}.} \bibinfo{year}{1994}\natexlab{}.
\newblock \showarticletitle{On the number of digital convex polygons inscribed into an (m,m)-grid}.
\newblock \bibinfo{journal}{\emph{IEEE Transactions on Information Theory}} \bibinfo{volume}{40}, \bibinfo{number}{5} (\bibinfo{year}{1994}), \bibinfo{pages}{1681--1686}.
\newblock
\urldef\tempurl%
\url{https://doi.org/10.1109/18.333894}
\showDOI{\tempurl}


\bibitem[Kakade and Foster(2008)]%
        {kakade2008deterministic}
\bibfield{author}{\bibinfo{person}{Sham~M. Kakade} {and} \bibinfo{person}{Dean~P. Foster}.} \bibinfo{year}{2008}\natexlab{}.
\newblock \showarticletitle{Deterministic calibration and Nash equilibrium}.
\newblock \bibinfo{journal}{\emph{J. Comput. System Sci.}} \bibinfo{volume}{74}, \bibinfo{number}{1} (\bibinfo{year}{2008}), \bibinfo{pages}{115--130}.
\newblock
\showISSN{0022-0000}
\urldef\tempurl%
\url{https://doi.org/10.1016/j.jcss.2007.04.017}
\showDOI{\tempurl}
\newblock
\shownote{Learning Theory 2004}.


\bibitem[Kleinberg et~al\mbox{.}(2023)]%
        {kleinberg2023ucalibration}
\bibfield{author}{\bibinfo{person}{Robert Kleinberg}, \bibinfo{person}{Renato~Paes Leme}, \bibinfo{person}{Jon Schneider}, {and} \bibinfo{person}{Yifeng Teng}.} \bibinfo{year}{2023}\natexlab{}.
\newblock \showarticletitle{U-Calibration: Forecasting for an Unknown Agent}. In \bibinfo{booktitle}{\emph{Proceedings of Thirty Sixth Conference on Learning Theory}} \emph{(\bibinfo{series}{Proceedings of Machine Learning Research}, Vol.~\bibinfo{volume}{195})}, \bibfield{editor}{\bibinfo{person}{Gergely Neu} {and} \bibinfo{person}{Lorenzo Rosasco}} (Eds.). \bibinfo{publisher}{PMLR}, \bibinfo{pages}{5143--5145}.
\newblock
\urldef\tempurl%
\url{https://proceedings.mlr.press/v195/kleinberg23a.html}
\showURL{%
\tempurl}


\bibitem[Luce(1959)]%
        {luce1959individual}
\bibfield{author}{\bibinfo{person}{R.~Duncan Luce}.} \bibinfo{year}{1959}\natexlab{}.
\newblock \bibinfo{booktitle}{\emph{Individual Choice Behavior: A Theoretical analysis}}.
\newblock \bibinfo{publisher}{Wiley}, \bibinfo{address}{New York, NY, USA}.
\newblock


\bibitem[McFadden(1976)]%
        {mcfadden1976quantalchoice}
\bibfield{author}{\bibinfo{person}{Daniel~L. McFadden}.} \bibinfo{year}{1976}\natexlab{}.
\newblock \bibinfo{booktitle}{\emph{Quantal Choice Analysis: A Survey}}.
\newblock \bibinfo{publisher}{NBER}, \bibinfo{pages}{363--390}.
\newblock
\urldef\tempurl%
\url{http://www.nber.org/chapters/c10488}
\showURL{%
\tempurl}


\bibitem[McKelvey and Palfrey(1995)]%
        {mckelvey1995quantalresponse}
\bibfield{author}{\bibinfo{person}{Richard~D. McKelvey} {and} \bibinfo{person}{Thomas~R. Palfrey}.} \bibinfo{year}{1995}\natexlab{}.
\newblock \showarticletitle{Quantal Response Equilibria for Normal Form Games}.
\newblock \bibinfo{journal}{\emph{Games and Economic Behavior}} \bibinfo{volume}{10}, \bibinfo{number}{1} (\bibinfo{year}{1995}), \bibinfo{pages}{6--38}.
\newblock
\showISSN{0899-8256}
\urldef\tempurl%
\url{https://doi.org/10.1006/game.1995.1023}
\showDOI{\tempurl}


\bibitem[Noarov et~al\mbox{.}(2023)]%
        {noarov2023highdimensional}
\bibfield{author}{\bibinfo{person}{Georgy Noarov}, \bibinfo{person}{Ramya Ramalingam}, \bibinfo{person}{Aaron Roth}, {and} \bibinfo{person}{Stephan Xie}.} \bibinfo{year}{2023}\natexlab{}.
\newblock \bibinfo{title}{High-Dimensional Prediction for Sequential Decision Making}.
\newblock
\newblock
\showeprint[arxiv]{2310.17651}~[cs.LG]


\bibitem[Okoroafor et~al\mbox{.}(2023)]%
        {okoroafor2023faster}
\bibfield{author}{\bibinfo{person}{Princewill Okoroafor}, \bibinfo{person}{Robert Kleinberg}, {and} \bibinfo{person}{Wen Sun}.} \bibinfo{year}{2023}\natexlab{}.
\newblock \showarticletitle{Faster Recalibration of an Online Predictor via Approachability}.
\newblock \bibinfo{journal}{\emph{arXiv preprint arXiv:2310.17002}} (\bibinfo{year}{2023}).
\newblock


\bibitem[Qiao and Valiant(2021)]%
        {qiao2021stronger}
\bibfield{author}{\bibinfo{person}{Mingda Qiao} {and} \bibinfo{person}{Gregory Valiant}.} \bibinfo{year}{2021}\natexlab{}.
\newblock \showarticletitle{Stronger calibration lower bounds via sidestepping}. In \bibinfo{booktitle}{\emph{Proceedings of the 53rd Annual ACM SIGACT Symposium on Theory of Computing}} (Virtual, Italy) \emph{(\bibinfo{series}{STOC 2021})}. \bibinfo{publisher}{Association for Computing Machinery}, \bibinfo{address}{New York, NY, USA}, \bibinfo{pages}{456–466}.
\newblock
\showISBNx{9781450380539}
\urldef\tempurl%
\url{https://doi.org/10.1145/3406325.3451050}
\showDOI{\tempurl}


\bibitem[Roth(2022)]%
        {roth2022uncertain}
\bibfield{author}{\bibinfo{person}{Aaron Roth}.} \bibinfo{year}{2022}\natexlab{}.
\newblock \showarticletitle{Uncertain: Modern Topics in Uncertainty Estimation}.
\newblock  (\bibinfo{year}{2022}).
\newblock
\urldef\tempurl%
\url{https://www.cis.upenn.edu/~aaroth/uncertainty-notes.pdf}
\showURL{%
\tempurl}


\bibitem[Roth(2023)]%
        {roth2023learning}
\bibfield{author}{\bibinfo{person}{Aaron Roth}.} \bibinfo{year}{2023}\natexlab{}.
\newblock \showarticletitle{Learning in Games and Games in Learning}.
\newblock  (\bibinfo{year}{2023}).
\newblock
\urldef\tempurl%
\url{https://www.cis.upenn.edu/~aaroth/GamesInLearning.pdf}
\showURL{%
\tempurl}


\bibitem[Zhao et~al\mbox{.}(2021)]%
        {zhao2021calibrating}
\bibfield{author}{\bibinfo{person}{Shengjia Zhao}, \bibinfo{person}{Michael Kim}, \bibinfo{person}{Roshni Sahoo}, \bibinfo{person}{Tengyu Ma}, {and} \bibinfo{person}{Stefano Ermon}.} \bibinfo{year}{2021}\natexlab{}.
\newblock \showarticletitle{Calibrating predictions to decisions: A novel approach to multi-class calibration}.
\newblock \bibinfo{journal}{\emph{Advances in Neural Information Processing Systems}}  \bibinfo{volume}{34} (\bibinfo{year}{2021}), \bibinfo{pages}{22313--22324}.
\newblock


\end{thebibliography}

\appendix
\newcommand{\calerr}{\textsf{CalErr}}

\section{Calibration and Swap Regret} \label{app:calibration}

Here, we discuss the relationship between calibration and agent swap regret, and connections to our results. For the remainder of this section, we focus on the 1-dimensional, binary prediction problem --- i.e. the outcome space is $\cC = \{0, 1\}$, and predictions $p_t$ take values in $[0, 1]$.

\subsection{Definitions and Best-Known Bounds for Calibration}

We first define the $\ell_1$ version of calibration error. 

\begin{definition}[$\ell_1$-Calibration Error]
    Given a transcript $\pi_T$ of predictions $\hat{y}_1,...,\hat{y}_T$ and outcomes $y_1,...,y_T$, the $\ell_1$-calibration error is:
    \[
    \calerr_1(\pi_T) = \sum_{y\in[0,1]} n_T(y) |y - \Tilde{y}_T|
    \]
    where $n_T(y) = \sum_{t=1}^T \1[\hat{y}_t = y]$ is the number of times $y$ was predicted, and $\Tilde{y}_T = \frac{\sum_{t=1}^T \1[\hat{y}_t = y] y_t}{n_T(y)}$ is the average outcome conditioned on the prediction being $y$.  
\end{definition}

\citet{foster1998asymptotic} give an algorithm achieving $\ell_1$-calibration error $O(T^{2/3})$. As a lower bound, \citet{qiao2021stronger} show that it is impossible to achieve $\ell_1$-calibrated forecasts at a rate of $\Tilde{O}(\sqrt{T})$. 

\begin{theorem}\citep{qiao2021stronger}
For any forecasting algorithm, there is an adversary such that the algorithm incurs $\ell_1$-calibration error $\Omega(T^{0.528})$. 
\end{theorem}

The $\ell_2$ variant of calibration error does not suffer the same difficulty. 

\begin{definition}[$\ell_2$-Calibration Error]
    Given a transcript $\pi_T$ of predictions $\hat{y}_1,...,\hat{y}_T$ and outcomes $y_1,...,y_T$, the $\ell_2$-calibration error is:
    \[
    \calerr_2(\pi_T) = \sum_{y\in[0,1]} n_T(y) (y - \Tilde{y}_T)^2
    \]
\end{definition}

\begin{theorem}\citep{roth2022uncertain}
    There is a forecasting algorithm that achieves $\Tilde{O}(\sqrt{T})$ $\ell_2$-calibration error.
\end{theorem}

\subsection{Relationship to Swap Regret}

Calibration is related to swap regret in two directions: low calibration error implies low swap regret for all downstream agents, and low swap regret for all downstream agents implies low calibration error. However, each direction of this connection uses a different measure of calibration error which is the reason why the connection is not tight.

Concretely, it is known that $\ell_1$-calibrated predictions lead to no swap regret for best-responding agents. 

\begin{theorem}\citep{foster1998asymptotic, haghtalab2023calibrated, kleinberg2023ucalibration}
For any transcript $\pi_T$ of predictions and outcomes, and any agent with action set $\cA$ and utility function $u$ who, at every round $t\in[T]$, best responds to $\hat{y}_t$, we have:
\[
\max_{\phi: \cA \to \cA} Reg(u, \phi) \leq \frac{4\calerr_1(\pi_T)}{T}
\]
\end{theorem}

In addition, low $\ell_2$-calibration error is necessary to guarantee all downstream agents low swap regret; in particular, it is possible to design an agent such that their swap regret is exactly the $\ell_2$-calibration error of our predictions.

\begin{theorem}\citep{kleinberg2023ucalibration} \label{ell_2 cal}
    Given any transcript $\pi_T$ of predictions and outcomes, there exists an agent with action set $\cA$ and utility function $u$ such that if, at every round $t\in[T]$, they best respond to $\hat{y}_t$, we have:
    \[
    \max_{\phi: \cA \to \cA} Reg(u, \phi) \geq \frac{\calerr_2(\pi_T)}{T}
    \]
\end{theorem}

In this paper we show how to minimize agent swap regret without minimizing $\ell_1$-calibration; our results rely instead on $\ell_\infty$-type bounds on conditional bias over best response events. It is worth mentioning that, while this technique allows us to consider the more general setting of real-valued, higher-dimensional outcome spaces, we acquire a dependence on the number of actions $|\cA|$ of the agent that renders their swap regret guarantee meaningless if $|\cA|$ is large, i.e. $\Omega(\sqrt{T})$ (this is the case for the agent constructed in Theorem \ref{ell_2 cal}). An important open question is if we can achieve \textit{action-independent} swap regret bounds for all downstream agents. 

For binary prediction, it turns out the answer is yes; \citet{hu2024predict} show in this case how to remove our dependence on the action space and obtain an action-independent $\Tilde{O}(\sqrt{T})$ guarantee for agent swap regret. 

\begin{theorem}\citep{hu2024predict}
    There is a forecasting algorithm such that any agent who best responds according to action set $\cA$ and utility function $u$ has swap regret bounded by:
    \[
    \max_{\phi: \cA \to \cA} Reg(u, \phi) \leq \Tilde{O}(\sqrt{T})
    \]
\end{theorem}

In other words, guaranteeing no swap regret for downstream agents is as hard as producing $\ell_2$-calibrated forecasts (which we can achieve at $\sqrt{T}$ rates). It is \emph{not}  as hard as producing $\ell_1$-calibrated forecasts, for which there exists a lower bound precluding $O(\sqrt{T})$ rates. This distinction between the rates achievable for $\ell_1$ and $\ell_2$ calibration is what resolves a seeming contradiction in the results of \citet{kleinberg2023ucalibration} and the results of our paper and \citet{hu2024predict}. 
\section{Algorithm for Unbiased Prediction}
\label{app:alg}

Here, we present the forecasting algorithm of \citet{noarov2023highdimensional} that gives guarantees of conditional unbiasedness as stated in Theorem \ref{thm:biasbound}. Recall that given a collection of events $\cE$, the goal is to produce predictions such that the bias conditional on every event $E\in\cE$ is low. 

The algorithm is simple to state and can be decomposed into two parts. The first part reduces to an experts learning problem---the algorithm computes a set of weights $q_t \in \Delta [2d|\cE|]$ by calling an experts algorithm, feeding in as losses the cumulative event biases. In the second part, it chooses $p_t$ by solving a minmax problem in which the objective is a convex combination of event biases, with coefficients given by $q_t$. We note that in the first part, any experts algorithm can be used in a black-box manner; for simplicity, we present the algorithm instantiated with the Exponential Weights algorithm with learning rate $\eta$. The theorems of \citet{noarov2023highdimensional} follow from instantiating the algorithm with more sophisticated no regret learning algorithms. Note that the $\min\max$ problem involves a very high dimensional action space for the minimization player, and so it is not a-priori clear how to solve it efficiently. However, \citet{noarov2023highdimensional} show that it can always be solved efficiently by taking advantage of the linear structure of the loss function in the maximization player's action, and the fact that the game has value 0. \citet{noarov2023highdimensional} show how to solve this minmax problem by having in rounds the maximization player choose their $d$-dimensional action using an online linear optimization algorithm, and having the minimization player ``best respond'' by copying the maximization player's strategy. 

\ifarxiv
\begin{algorithm}
\caption{Algorithm for Unbiased Prediction}
\begin{algorithmic}
\For{$t = 1...T$}
    \State Define the distribution $q_t \in \Delta [2d|\cE|]$ such that for $E \in \cE, i\in[d], \sigma\in \{\pm 1\}$,
    \[
    q_t^{E, i, \sigma} \propto \exp\left( \frac{\eta}{2} \sum_{s=1}^{t-1} \sigma \cdot \E_{\hat{y}_s \sim p_s}[E(\hat{y}_s) (\hat{y}_s^i - y_s^i)] \right)
    \]
    \State Output the solution to the minmax problem:
    \[
    p_t \gets \argmin_{p_t' \in \Delta\cC} \max_{y \in \cC} \E_{\hat{y}_t \sim p_t'}\left[\sum_{E, i, \sigma} q_t(E, i, \sigma) \cdot \sigma \cdot E(\hat{y}_s) \cdot (\hat{y}_s^i - y_s^i) \right]
    \]
\EndFor
\end{algorithmic}
\end{algorithm}

\else
\begin{algorithm}
\caption{Algorithm for Unbiased Prediction}
    \For{$t=1...T$}{
    Define the distribution $q_t \in \Delta [2d|\cE|]$ such that for $E \in \cE, i\in[d], \sigma\in \{\pm 1\}$,
    \[
    q_t^{E, i, \sigma} \propto \exp\left( \frac{\eta}{2} \sum_{s=1}^{t-1} \sigma \cdot \E_{\hat{y}_s \sim p_s}[E(\hat{y}_s) (\hat{y}_s^i - y_s^i)] \right)
    \]

    Output the solution to the minmax problem:
    \[
    p_t \gets \argmin_{p_t' \in \Delta\cC} \max_{y \in \cC} \E_{\hat{y}_t \sim p_t'}\left[\sum_{E, i, \sigma} q_t(E, i, \sigma) \cdot \sigma \cdot E(\hat{y}_s) \cdot (\hat{y}_s^i - y_s^i) \right]
    \]
    }
\end{algorithm}
\fi

\end{document}